\documentclass[pra,aps,groupedaddress,twocolumn,floatfix,nofootinbib]{revtex4-2}

\pdfoutput=1

\usepackage{fullpage}

\usepackage{ulem}
\usepackage{amsmath}
\usepackage{pifont}
\usepackage{amsthm}
\usepackage{amsthm}
\usepackage{cancel}
\usepackage{changes,subcaption,comment}
\usepackage{paralist}
\usepackage{amssymb}
\usepackage{hyperref}
\usepackage{pifont}
\usepackage[margin=1in]{geometry}
\usepackage{array}
\usepackage{colortbl}
\usepackage{xcolor}
\usepackage{enumitem}

\newcommand{\blk}{\color{black}}

\usepackage{enumitem}
\setlist[itemize]{noitemsep}
\setlist[enumerate]{noitemsep}

\begin{document}

\title{Superluminal Transformations and Indeterminism}

\author{Amrapali Sen}
\thanks{Corresponding author}
\affiliation{International Centre for Theory of Quantum Technologies, University of Gda\'{n}sk, Gda\'{n}sk, Poland}

\author{Flavio Del Santo}
\affiliation{Group of Applied Physics, University of Geneva, 1211 Geneva 4, Switzerland;\\ and Constructor University, Bremen, Germany}

\date{\today}

\begin{abstract}

Quantum theory is widely regarded as fundamentally indeterministic, yet classical frameworks can also exhibit indeterminism once infinite information is abandoned. At the same time, relativity is usually taken to forbid superluminal signalling, although Lorentz symmetry formally admits superluminal transformations (SpTs). Dragan and Ekert have argued that SpTs entail a form of indeterminism analogous to that encountered in quantum theory. Here, we derive a theory-independent no-go theorem from a set of natural assumptions: any framework admitting non-order-preserving SpTs must either abandon finite information, relinquish time-symmetric informational content, deny that the past stores memory, or abandon the notion that time determines a preferred causal ordering. In particular, one possible implication is that any theory accommodating SpTs suggests an ontology with unbounded informational content, akin to deterministic classical theories formulated over the real numbers. Consequently, any ontic indeterminacy associated with superluminal transformations \textit{cannot} originate from finite information.

\end{abstract}
\maketitle
\section{Introduction}
\label{intro}
Quantum theory is often taken to exhibit a form of fundamental indeterminacy that manifests both in the absence of perfectly determinate properties (ontic indeterminacy) and in the non-uniqueness of future outcomes associated with a given pure state (fundamental indeterminism). This is supported by the Heisenberg uncertainty principle, which limits the simultaneous determination of conjugate quantities, and by violations of Bell inequalities, which, under natural assumptions, imply that the outcomes of Bell experiments cannot be predetermined \cite{pironio2010random}.

More recently, several works have challenged the view that classical theories are fully deterministic \cite{born2012physics, ornstein1989ergodic, prigogine1997end, de2009symplectic, norton2003causation, drossel2015relation, del2018striving, del2019physics, gisin2020real, eyink2020renormalization, ben2020structure, del2021relativity, gisin2021indeterminism, gisin2021synthese, del2021indeterminism, del2023prop, santo2023open, chiribella, DelSantoGisin2024, del2025features}. One line of argument, initiated by N.~Gisin and one of us (FDS) \cite{del2019physics}, maintains that classical physics exhibits ontic indeterminacy once a tacit assumption is abandoned---namely, that physical quantities are specified with infinite information. Rejecting this assumption amounts to discarding real-number representations, most of which are uncomputable and encode infinite informational content \cite{gisin2020real, gisin2021indeterminism}. By relaxing the requirement of infinite precision, indeterminacy is no longer uniquely quantum but can also be intrinsic to classical physics. Combined with chaotic dynamics, this opens the possibility that classical systems inherently display indeterministic behaviour \cite{del2019physics, del2021indeterminism, del2025features}.

At the same time, Dragan and Ekert \cite{dragan2020quantum} recently employed superluminal transformations (SpTs), first introduced in \cite{parker1969}, extending the Lorentz group to include superluminal velocities, to argue that such transformations give rise to indeterminism in quantum physics. In this way, they aim to ground quantum indeterminism in relativistic superluminal considerations. In a previous work, one of us (AS) analyzed superluminal quantum reference frames for thermodynamic and probabilistic consistency \cite{sen2026superluminal}.

%Rather than supporting Dragan and Ekert's claim, we argue that once finite-information assumptions are adopted, indeterminism already arises independently of superluminal considerations. We therefore formulate a no-go result \textcolor{olive}{for SpTs based on finite distributions of information in the direction of time.} \sout{against grounding indeterminism in SpTs.} 

%\textcolor{magenta}{maybe we skip the next paragraph, also not sure we want to say this because the fundamental indeterminacy in classical theories is because of the constraint of finite information,}
%\blue{\sout{In this work, we connect indeterminism in these two distinct settings: frameworks based on finite informational content of physical states, and those in which indeterminism is attributed to superluminal transformations (SpTs), as suggested in \cite{dragan2020quantum}. In particular, we prove a theory-independent no-go theorem demonstrating that SpTs cannot serve as a fundamental origin of indeterminism when physical states are constrained by finite information in the direction of time}\sout{Provided the other assumptions hold, this shows that their joint implementation leads to a structural inconsistency.}}

In contrast, here we formulate a theory-independent no-go theorem which, under a set of natural assumptions, establishes an incompatibility between non-order-preserving SpTs and finite information. This suggests that the indeterminacy characteristic of finite-information frameworks cannot straightforwardly be reconciled with SpTs.

We analyze the structure of SpTs, with particular attention to the subclass of Superluminal Lorentz Transformations (SpLTs) \cite{parker1969}, and their implications for the distribution of finite information in spacetime. Our arguments do not depend on the specific form of the SpLTs introduced in \cite{parker1969} and studied in \cite{dragan2020quantum, sen2026superluminal}; they apply to a broad class of SpTs, including nonlinear generalizations, provided that they are non-order-preserving for timelike events. Such transformations reverse the temporal order of at least one set of spacetime events with respect to the initial reference frame.

%\blue{\sout{This proposal should be understood as a structural no-go result concerning two general features of hypothetical theories, namely, fundamental indeterminacy and superluminal transformations}} \textcolor{magenta}{it might be sounding like we are saying fundamental indeterminacy and SpT do not go together, but this is not what we claim, maybe we dont write this sentence}. 
Our analysis does not rely on the details of any specific physical model --- indeed, we do not claim that superluminal theories correspond to physical reality --- but rather explores, in the spirit of a toy framework, generic structural features shared by broader classes of theories exhibiting these properties. Although speculative, this approach may help clarify the conceptual structure of indeterministic theories and question the claim that superluminal transformations could constitute the fundamental origin of indeterminacy, as suggested in~\cite{dragan2020quantum}.

The possibility of superluminal phenomenons and observers has traditionally been challenged, primarily because such entities are often argued to permit backwards-in-time signalling and therefore generate causal paradoxes \cite{einstein1907uber,tolman1917velocities}. Nevertheless, the status of superluminal phenomenon has been frequently questioned, and the idea of breaking this speed limit has popped up from time to time, both from a purely theoretical point of view as well as in an attempt to explain various phenomena \cite{bilaniuk1962meta,feinberg1967possibility,recami1986classical,rembielinski2012meta,grushka2013tachyon}. Regardless, these research programs are interesting from a number of perspectives. If one takes the position that superluminal phenomena has genuine physical significance, their study would naturally carry explanatory importance. But even if one’s point of view is opposed to this idea, these frameworks continue to serve as interesting conceptual tools. Understanding why exactly they work or fail can potentially yield new insights into how different principles of a theory constrain each other. %Along similar lines, in order to confirm our current theories we need to check that what they they deem impossible is actually so.

The paper is organized as follows. Section II revisits classical indeterminism from the standpoint of finite information. Section III reviews indeterminism arising from superluminal transformations. Section IV outlines and reviews the information-theoretic framework for propensities and indeterministic causality \cite{santo2023open}. Section V presents our main result: a no-go theorem showing that superluminal transformations and finite information are incompatible (given a few other natural assumptions). Section VI discusses the interpretational implications of relaxing each assumption. Section VII gives a discussion of the results and our interpretation of indeterminacy from superluminal transformation, and some future perspectives. Section VIII concludes.

\section{Classical indeterminism from finite information}

%\textcolor{gray}{(It’s not immediately clear to me that this limitation of real numbers is ontic, rather than epistemic. Also,
%classical physics can only provide an approximate description of physical reality, so classical arguments can only
%have a limited applicability. I at least agree that practical classical calculations are limited by this information
%limit criterion. But again, this is an epistemic problem for the classical physicist, not necessarily a problem for
%classical matter (which is an idealization).
%Are the authors arguing that information density is the source of quantum indeterminacy, or do they just make
%a point about classical physics?} \textcolor{magenta}{We can look at this problem at the end}
The orthodox interpretation of classical theory regards it as fully deterministic. However, this conclusion can only be drawn from two joint assumptions: (i) at the dynamical level, the solutions to the equations of motion in classical physics---as represented by ordinary partial differential equations---admit a unique continuation both into the past and into the future; provided that (ii) at the kinematical level, the initial conditions are infinitely precise---as represented by real numbers in the theory.

In a recent series of works \cite{del2019physics, gisin2020real, del2021relativity, gisin2021indeterminism, gisin2021synthese, del2021indeterminism, del2023prop, santo2023open, DelSantoGisin2024, del2025features}, N.~Gisin and one of us (FDS) have investigated the consequences of relaxing assumption (ii) by adopting a principle of finite information density. According to this principle, a finite region of space (or spacetime) can contain only a finite amount of information. Here we adopt a physicalist view of information: it resides in the distinct states of the universe’s physical degrees of freedom. This can be quantified in different ways; for instance, by Kolmogorov complexity, which measures the information content (in bits) of the shortest algorithm that outputs the considered piece of information.\footnote{In what follows, unless otherwise specified, we will use Kolmogorov complexity as the measure of information. Note, however, that our results do not depend on the specific choice of information measure.}

%\textcolor{gray}{The authors correctly include chaotic behavior as part of the dynamics, which is noted in the Introduction on page 1 and at the end of Section II on page 2. However, chaotic behavior and the corresponding predictions should be discussed in greater depth within the current manuscript.}
%\textcolor{red!80!black}{Can write sth later, unless, Flavio?} 

Since most real numbers are incomputable and carry infinite information (their Kolmogorov complexity is infinite), the use of real-valued $n$-tuples to represent pure classical states is incompatible with this principle. In Ref.~\cite{del2019physics}, this issue was addressed by replacing real numbers with \textit{finite information quantities} (FIQs), which contain only a finite amount of information at any finite time. FIQs are grounded in the notion of propensities \cite{del2023prop}: each digit of the mathematical representation of a physical quantity is associated with an objective tendency or potentiality, quantified by a rational number (between 0 and 1) to assume a particular value. Propensities then evolve in time, and when they take one of their extremal value, the associated digit becomes fully determined. This framework thus leads to a classical theory which features ontic indeterminacy, but that is empirically equivalent to the standard deterministic interpretation.

%\textcolor{gray}{Reading the 2019 paper:
%- What is the `time’ over which propensities are active to determine bit values? \textcolor{blue}{The “time” relevant for the activation of propensities should be understood as geometric(coordinate) time associated with the dynamical evolution of the system, consistent with standard formulations in Special Relativity}- Where does the information to fill in the later bits come from? \textcolor{blue}{Maybe the later bits emerge from the dynamical unfolding of the system, including nonlinear interactions, sensitivity to initial conditions, and the progressive formation of correlations. In this sense, the process is not purely random, but reflects the evolution of a system with finite informational resolution.} It is randomly generated? Is it related to
%Energy – Time uncertainty?}
%\textcolor{blue}{In quantum theory, such limitations can be expected to be tied to energy scales, whereas in the present framework it appears as a more general bound on the amount of information that can be resolved or generated in finite time.}\textcolor{red!80!black}{The uncertainity in Quantum theory is theory dependent problem? And our toy model is thery indepenedent?} \textcolor{red!80!black}{Flavio?}
We stress that, although in that particular model the dynamical part of classical theory (i.e., the equations of motion) was left unchanged, the ontic indeterminacy of the states alone is sufficient to introduce genuine indeterminism. Classical dynamics may exhibit chaotic behaviour, characterized by an extreme sensitivity to initial conditions: infinitesimal differences between nearby states grow exponentially with time, on a characteristic timescale set by the Lyapunov time. In deterministic classical mechanics, this feature is often regarded as merely epistemic, since perfectly specified initial conditions still determine a unique trajectory. In the present framework, however, the initial indeterminacy is assumed to be ontic (similar, though for different reasons, to the indeterminacy of a pure quantum state). Under chaotic evolution, this finite indeterminacy is dynamically amplified: although Liouville evolution preserves the total phase-space measure, the support of the state becomes progressively stretched and dispersed across macroscopically distinct regions of phase space. After sufficiently many Lyapunov times, the evolved state therefore corresponds to multiple physically possible future evolutions rather than to a single predetermined history, giving rise to genuine indeterminism.

It should be noted that the present proposal does not depend on any specific interpretation of classical indeterminism; we refer to such models only because they are already well developed in the literature. The argument applies more generally to any theory in which indeterminacy arises from a finite informational content in the specification of the state, including possible finitary formulations of quantum theory. We are therefore not concerned with particular existing theories, but rather with a structural comparison between toy frameworks that share these fundamental features.

%%%%%%%%%%%%%%%%%%%%%%%%%%%%%%%%%%%%%%%%%%%%%%%%%%%%%%%

%\textcolor{gray}{Comment: I do agree that initial conditions / chaotic dynamics is a potential source of classical indeterminacy.
%The idea of limiting classical information (FIQs) does somehow accord with the idea of limits we know exist
%from quantum. Both limit information in their own ways. The classical world with FIQ limits is a model.
%“finite information at a finite time” also accord with quantum Energy-Time uncertainty. This may suggest a
%connection between FIQ and energy?}

\section{Indeterminism from Superluminal transformations}

%\textcolor{gray}{My understanding: I am not 100 percent clear on the logic here. If they “give up local deterministic dynamics” is this
%the step saying that for the superluminal regime to avoid superluminal signalling, they are forced to postulate
%indeterministic dynamics? So, the indeterministic mechanism saves the normal ordering of cause-and-effect? \textcolor{blue}{it delocalised enough to bring it back inside the lightcone in the new reference frame}
%Is this being presented as a prediction? \textcolor{blue}{I am writing some stuff up with Vialsini where we do these, maybe an mention?}It is a necessary condition? I suppose the argument could go either
%way? It could be taken as either evidence against superluminal, or for indeterminism.}\textcolor{red!80!black}{This last sentence is what I don't understand at all.}
 Special relativity is typically interpreted as enforcing an absolute prohibition on superluminal motion: nothing, whether matter or information, may exceed the speed of light. The conventional justification is that faster-than-light propagation would enable signalling to the past and thereby undermine causal consistency. Because this argument is so familiar, it is often treated as definitive. However, the claim that superluminal transformation is fundamentally impossible has never been conclusively established, and it tends to re-emerge whenever the limits of established theory are reconsidered.

 The proposal put forward in Ref.~\cite{dragan2020quantum} exemplifies this continued debate. The authors advocate a reversal of the usual interpretation: rather than viewing superluminal observers as incompatible with relativity, they argue that the real source of difficulty lies in the assumption of local determinism, linking the two pillars of physics---quantum theory and relativity---and suggesting that the latter provides the foundation for the former. According to their analysis, once this assumption is abandoned, the standard paradoxes associated with superluminal motion cease to arise, since with sufficiently large uncertainties observers\footnote{We adopted the standard usage from quantum theory in which “observation” is understood operationally as a measurement process, rather than implying any role for a conscious observer. More precisely, a measurement corresponds to a physical interaction that establishes correlations between a system and another physical system that records the outcome. In this sense, an “observer” is simply any physical system that can encode or register such information, not necessarily a human agent.} %\textcolor{gray}{Additionally, the authors should clearly specify whether they use the terms "observation" or "observer" to describe measurement. This is particularly relevant in lines 42-47 of the right column on page 2.}\textcolor{red!80!black}{I feel like both of them mean the same here, what am I missing?} 
 cannot establish whether they have in fact observed superluminal signalling. On this basis, the authors further suggest that the intrinsic randomness characteristic of quantum theory is not merely compatible with relativity but may, in fact, arise from it. Unsurprisingly, this position has generated significant criticism in the literature~\cite{DelSanto_2022, Horodecki2023, grudka2022, Grudka2023}. The mathematical derivation of Lorentz transformations that the authors of Ref. \cite{dragan2020quantum} talk about contains both subluminal and superluminal velocities as derived in \cite{parker1969}, which we will now briefly review.

Let us consider the (1+1)-dimensional case, where an inertial frame \((t', x')\) moves with a constant velocity \(v\) relative to another inertial frame \((t, x)\). By the principle of relativity, the transformation between these two frames must be linear and depend only on the relative velocity \(v\). Moreover, the inverse transformation should be obtained simply by reversing the sign of \(v\). Hence, we can write the general form of such transformations as
\begin{gather}
\begin{aligned}
x' &= A(v)\,x + B(v)\,t,\\
x  &= A(-v)\,x' + B(-v)\,t',
\end{aligned}
\end{gather}
where \(A(v)\) and \(B(v)\) are functions to be determined.

%To relate these functions, consider a point at rest in the primed frame, for which \(x' = 0\). Substituting into the first equation gives \(A(V)\,x + B(V)\,t = 0\), which implies
%\begin{equation}
%\frac{B(V)}{A(V)} = -\frac{x}{t} = -V.
%\end{equation}
%Thus,
%\begin{equation}
%B(V) = -V\,A(V).
%\end{equation}
%This relation expresses that a point at rest in one frame moves with velocity \(V\) in the other.

Next, we note that reversing the spatial axis \(x \to -x\) also reverses the relative velocity \(v \to -v\). Since the transformation is linear, this reversal should preserve the form of the equations, possibly up to an overall sign. This condition implies that \(A(v)\) must be either symmetric,
\begin{equation}
A(-v) = A(v),
\end{equation}
or antisymmetric,
\begin{equation}
A(-v) = -A(v).
\end{equation}
%The symmetry of \(A(V)\) in turn determines the parity of \(B(V)\) through the relation \(B(V) = -V\,A(V)\).

%Substituting Eq.~(3) into Eq.~(1) gives
%\begin{gather}
%\begin{aligned}
%x' &= A(V)(x - Vt),\\
%x  &= A(-V)(x' + Vt').
%\end{aligned}
%\end{gather}
%Requiring consistency between these two equations leads to the condition
%\begin{equation}
%A(V)A(-V)\,[1 - V^2] = 1.
%\end{equation}
Depending on the symmetry of \(A(v)\), this yields two distinct solutions.

For the symmetric case, \(A(-v) = A(v)\), we obtain
%\begin{equation}
%A(V)^2(1 - V^2) = 1,
%\end{equation}
%or equivalently,
%\begin{equation}
%A(V) = \frac{1}{\sqrt{1 - V^2}}.
%\end{equation}
%Using \(B(V) = -V\,A(V)\), the resulting transformation corresponds to the usual Lorentz transformation (setting \(c = 1\)):
\begin{gather}
\begin{aligned}
x' &= \frac{x - vt}{\sqrt{1 - v^2}},\\
t' &= \frac{t - vx}{\sqrt{1 - v^2}}.
\end{aligned}
\end{gather}
This form is well-defined for subluminal velocities \(|v| < 1\) (where c=1). \footnote{At this point, several earlier approaches attempting to extend relativistic kinematics beyond the speed of light introduced complex or imaginary coordinate structures, including notions of “imaginary time” or imaginary Lorentz factors, in order to formally accommodate superluminal velocities: \cite{parker1969ftl,ramon1980complex,maccarrone1984revisitation,maccarrone1982group,ibison2007tachyons,roldan2023superluminal}, to name a few. In many of these formulations, the appearance of imaginary quantities does not correspond to directly measurable temporal evolution, but rather signals a transition into regimes that are inaccessible within ordinary classical spacetime descriptions. In the present work, we avoid interpreting superluminal sectors through imaginary coordinate constructions inherited from subluminal Lorentz transformations and focus on independent superluminal transformations which can be derived directly from mathematical constraints in the basic principles of special relativity. Our aim is therefore not to analytically continue the subluminal theory into complex domains, but rather to investigate the structural and informational consequences of admitting superluminal transformations on their own conceptual footing, independent of their subluminal counterparts.}

In contrast, for the antisymmetric case \(A(-v) = -A(v)\), we obtain: %the consistency condition gives
%\begin{equation}
%-\,A(V)^2(1 - V^2) = 1,
%\end{equation}
%leading to
%\begin{equation}
%A(V) = \pm\,\frac{1}{\sqrt{V^2 - 1}}.
%\end{equation}
%Substituting this into the transformation equations yields
\begin{gather}
\begin{aligned}
x' &= \pm\,\frac{v}{|v|}\,\frac{x - vt}{\sqrt{v^2 - 1}},\\
t' &= \pm\,\frac{v}{|v|}\,\frac{t - vx}{\sqrt{v^2 - 1}}\label{eq: sup}.
\end{aligned}
\end{gather}

%\textcolor{gray}{A brief discussion regarding Expression 4 and its reference to imaginary times—such as in quantum tunneling—could be quite engaging.}\textcolor{magenta}{i guess? since both are some sort of reparametrisation}
These transformations are well-behaved for superluminal velocities \(v > 1\). Since there is no \(v \to 0\) limit in this case, the overall sign remains undetermined and may be fixed by convention. Also, note that at superboost, i.e., $v \rightarrow \infty$, space and time coordinates swap, i.e., $x' \rightarrow t$; $t'\rightarrow x $.

Thus, by assuming linearity, relativity, and invariance under spatial reversal, we recover both the standard Lorentz transformations for subluminal motion and a consistent ``tachyonic'' extension valid for superluminal motion.

We would also like to stress again that the transformation (SpLT) derived above is only an example of a superluminal transformation that our argument applies to. The no-go theorem we propose later holds equally well for any form of SpTs if we ensure non-order-preserving transformations for at least some timelike events. Also, the arguments put forward in this paper about SpTs can be extended to any general non-order-preserving transformations.

Many discussions surrounding superluminal phenomena implicitly blur the distinction between superluminal signalling, superluminal causation, and the existence of superluminal Lorentz transformations. These notions, however, are conceptually distinct. Superluminal signalling is established impossible, but the study of superluminal causal structures has been explored extensively in several contexts \cite{vilasini2022impossibility,horodecki2019relativistic,durr2009bohmian}. 

In ref. \cite{dragan2020quantum}, the authors claim that the inherent randomness we get in quantum theory is a result of the special theory of relativity. They suggest that the overlooked superluminal Lorentz transformations derived in \eqref{eq: sup} imply non-deterministic initial conditions and dynamics giving rise to the same ontic randomness that quantum theory brings. However, their arguments with superluminal Lorentz transformation respect the postulates of relativity and no-superluminal signalling: no information travels faster than the speed of light. The relevant indeterminacy does not arise from incompatible observables, as in standard quantum theory, but from the impossibility of operationally establishing definite superluminal signalling relations under sufficiently large uncertainties.%They do so by giving up local deterministic dynamics and classically well-defined initial conditions\textcolor{magenta}{maybe we dont want to talk about initial conditions here} for both subluminal and superluminal observers
\footnote{Here superluminal observer could be a subliminal observer that is boosted to a superluminal velocity, moving now in a spacelike trajectory. It is the same as seeing a subluminal observer in a superluminally boosted reference frame.}.

\section{Propensities and Indeterministic Causality}

%\textcolor{magenta}{i will mention causal set theory approaches by sorkin and fay dowker etc} \red{Yes, please do it where is relevant.}

Following our earlier discussion, finite information imposes fundamental limits on how precisely physical properties can be determined, thereby leading to ontic indeterminacy. As recalled in Section~\ref{intro}, when such indeterminacy is combined with chaotic dynamics, it gives rise to a growing uncertainty regarding future states, corresponding to multiple possible outcomes of future observations, i.e., to genuine indeterminism. Importantly, although there remains a causal connection between past and future states, different future outcomes may possess different intrinsic likelihoods of occurring.

Such likelihoods are often represented by probabilities. However, in the literature on indeterminism in physics, both classical and quantum, one frequently encounters instead the notion of \textit{propensity} \cite{popper1959propensity, suarez2010probabilities, ballentine2016propensity, gisinpropold, del2023prop}. In Ballentine’s words, propensities constitute ``a form of causality that is weaker than determinism'' \cite{ballentine2016propensity}. More precisely, they are objective dispositional properties characterizing the tendency of a particular event to occur. In this sense, the propensity framework may be viewed as a generalization of the Leibnizian Principle of Sufficient Reason: nothing occurs without a reason or causal ground, while causation itself is understood more broadly than strict deterministic implication.

A central issue is explaining how a system with several possible outcomes ultimately resolves into one definite result. Prior to observation, the different outcomes exist only as possibilities with associated likelihood (i.e., propensity), yet measurement always yields a single realized state. This creates a tension between the continuous evolution of the system and the emergence of one concrete outcome. This transition is referred to as \textit{actualization} \cite{del2023prop}, extending the notion of collapse in quantum mechanics to any indeterministic framework involving a shift from potentiality to actuality. As with the quantum measurement problem, the underlying mechanism responsible for this transition—if such a mechanism exists—remains unresolved \cite{del2021indeterminism, del2023prop}.

For example, a propensity equal to $1$ corresponds to certainty that a given event will occur, whereas a propensity of $1/2$ expresses equal tendencies for occurrence and non-occurrence. A propensity of $3/4$ indicates that the event possesses a stronger intrinsic tendency to actualize than not to actualize, precisely quantified by that value.

Propensities are therefore closely related to probabilities, but they differ in several crucial aspects. First, they characterize single-case events, rather than frequencies emerging from repeated trials or subjective degrees of belief. Second, and more importantly, propensities are intended to encode an underlying causal structure rather than merely quantify statistical correlations.

In this context, it is useful to recall Eddington’s distinction between \textit{causality} and \textit{causation} \cite{eddington2019nature}. By \textit{causality}, one refers simply to the possibility that two events are causally connected---for instance, that they are not spacelike separated in relativistic terms. By \textit{causation}, instead, one refers to the stronger notion whereby an event $C$ genuinely acts as the cause of an effect $E$. Causation therefore introduces a fundamental asymmetry between events. Propensities aim precisely at quantifying the strength of such directed causal tendencies.

Because propensities encode genuine causal asymmetry, they cannot formally be identified with ordinary conditional probabilities, as emphasized by Humphreys \cite{humphreys1985propensities}. Indeed, standard conditional probabilities are formally symmetric under Bayes’ rule:
\begin{equation}
    p(E|C)=\frac{p(C|E)p(E)}{p(C)}.
\end{equation}
In such expressions, the conditioning relation is purely formal and may always be inverted mathematically. Propensities, however, are not merely formal labels; they embody an objective directional relation between cause and effect. For instance, there may exist a propensity for an individual to drink a glass of cold water given that the day is extremely hot. The reverse implication---that the hot weather occurs because the individual drinks cold water---simply does not represent a meaningful causal relation. The asymmetry is therefore physical and ontological, not merely epistemic, and should therefore lead to a formalization of propensities that reflects this desideratum \cite{ballentine2016propensity, del2023prop}. Note that this formalizes what is already a widespread physical intuition in quantum mechanics: a radioactive atom possesses an intrinsic tendency to decay, quantified by its quantum state. In this sense, propensities provide a conceptual framework capable of unifying classical and quantum forms of indeterminism while keeping causality at the centre of the picture.

Equipped with the notion of propensities, causal influence naturally defines a branching structure of possible futures. We assume this structure to be dynamical: as time evolves, only one among the various possible events or outcomes becomes actual. In other words, from a set of potentialities quantified by propensities, a single outcome is eventually actualized. There is therefore a genuine process of actualization of propensities, analogous to the transition from a superposition of possible outcomes to a single realized state in quantum mechanics during wavefunction collapse.
Importantly, this process should not be understood epistemically, as if one merely reveals a pre-existing but unknown outcome. Rather, the actualization process is ontological: prior to actualization, the different future possibilities genuinely exist only as potentialities endowed with different causal tendencies to occur.

In this view, it is convenient to describe the evolution  of a system using a directed acyclic multigraph where two distinct types of causal links connect events: potential (represented in blue) and actual (represented in red), see Fig. \ref{f3.2}. The ``potential" causal connections are associated with propensities—which are non-negative rational numbers that sum to 1 and quantify the objective tendency of an event to obtain. The ``actual" causal connections are associated with extremal propensities, which are either $0$ or $1$, with $1$ indicating that the actualization of one among the previously possible causal connections has taken place  \cite{santo2023open}.

\begin{figure}[h]
    \centering
    \includegraphics[width=0.8\columnwidth]{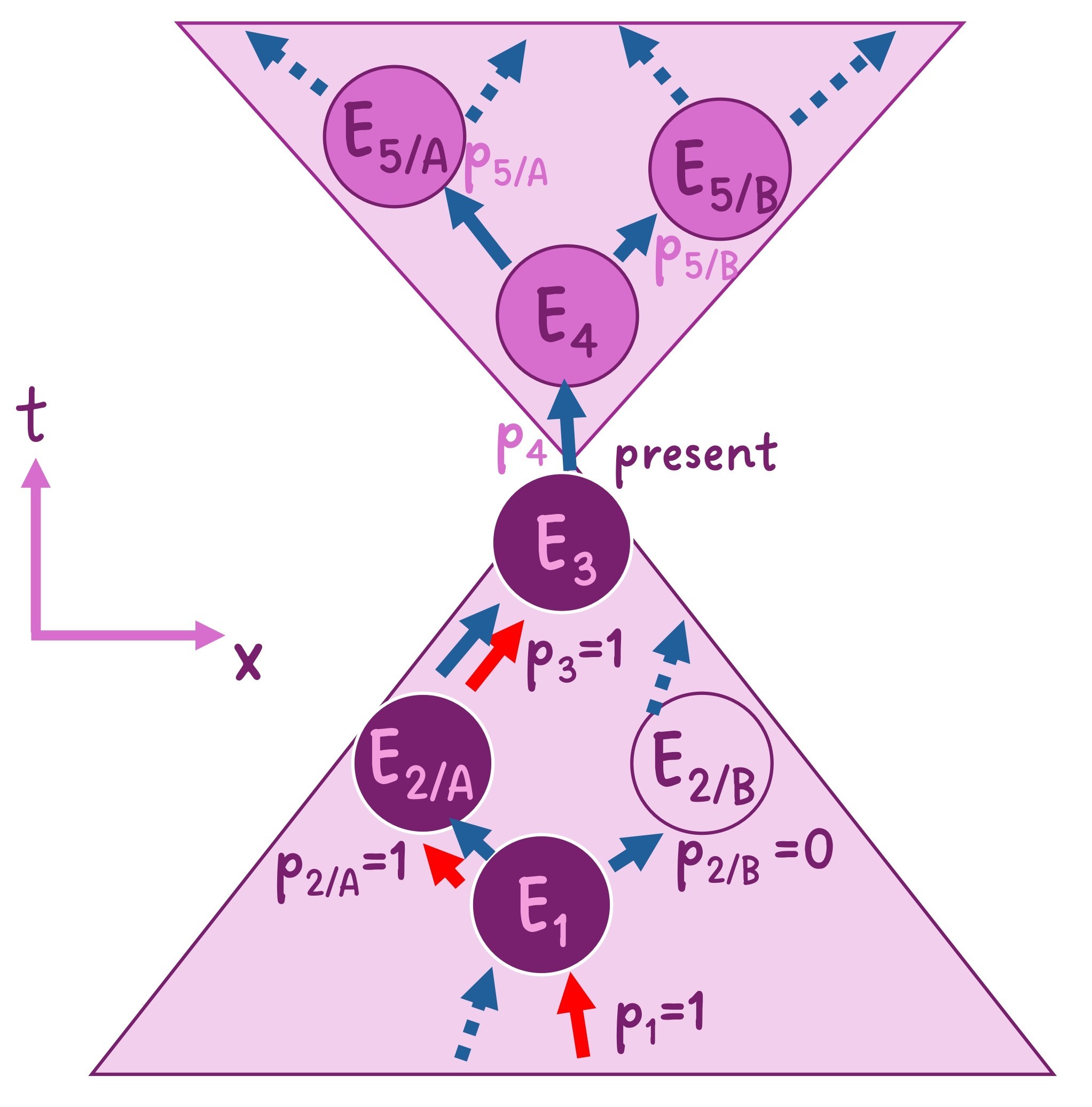}
    \caption{Past and Future of Information: The nodes represent events, the red edges connect events that have already actualized, while the blue edges represent potential causal connections, associated with a weight that represents the propensity for an event to actualize. The dotted edges going out of the information cones (which are bounded by the lightcone) of the multigraph are causal connections that either never actualize or do not actualize to observable interaction, hence, for our consideration, the propensity of such events is zero.}
    \label{f3.2}
\end{figure}

%\textcolor{gray}{Regarding Figure 1 on page 4: Does the present encompass only one possible event?
In the framework considered here, the “present” does not denote a single isolated event, but rather a boundary (or interface) between realized events (the past) and unrealized possibilities (the future) regarding a continuity equation, we do not assume a conserved current in
the usual dynamical sense at the level of this schematic representation. The figure is meant to
capture a kinematical structure rather than a fully specified dynamical law. Any continuity
relation linking past and future would therefore be theory-dependent and lies beyond the
scope of the present toy model. 

%Is there a continuity equation that links the past and future at the limiting point of the present?\textcolor{blue}{we just say it is dynamic?} Are events E2/A and E2/B in a state of superposition? Could they be entangled?} \textcolor{blue}{Yes, I guess, like QRFs, right? Not since we are keeping the manuscript theory independent maybe we dont want to commit} %\textcolor{gray}{The scenario described brings to mind Maxwell’s Demon; therefore, incorporating a brief discussion on entropy would certainly be worthwhile and interesting.}
%\textcolor{red!80!black}{not exactly sure what to write}

The state of a system can be defined as the information recording the whole multigraph,\footnote{This is comparable with causal modeling language in \cite{vilasini2022general} where this graph shows potential causal connections (the wavy causal arrows), potential cause and affects relations (change in probability = solid ones in which the wavy causal arrows actualizes), and then potential causal connections that do not actualize and hence do not satisfy affects relations- dashed arrows. The relations are summarized in Definition 7 and Fig. 7 in \cite{vilasini2022general}}\footnote{The present framework bears some conceptual resemblance to causal set approaches developed by Sorkin, Dowker, and collaborators \cite{Dowker2005, Sorkin2003, Bombelli1987}, in the sense that causal and informational relations between events are treated as more fundamental than a fixed spacetime geometry. However, unlike causal set theory, we do not assume that spacetime itself is fundamentally discrete or represented by a partially ordered set. Our framework instead remains kinematical and informational in character, focusing on how event structures and informational accessibility transform under extended Lorentz-type transformations.} which is an ordered collection
\begin{equation}
    G(t): [E(t), \Omega(t), P(t)],
\end{equation}

%\textcolor{gray}{Does expression (6) also account for creative time?}
%\textcolor{blue}{It does not explicitly encode “creative time” in the strong sense of generating genuinely new possibilities. Rather, it describes the evolution or weighting of events within a given set of admissible configurations, one of which actualizes spontaneously in creative time?}
  where $E(t)$ is the set of all events (nodes), $\Omega(t)$ is the set of all (potential and actual) causal connections (edges) which covers all possible alternative histories and futures,  $P(t)$ is the set of individual propensities (weights) all at the considered time $t$. The manifold of our causal graph encompasses all potential events (as nodes, however, represent all possible events). Some of these will never actualize. This is why we have a multigraph, where the different types of edges distinguish exactly this. 

  %\textcolor{gray}{Comment: Part 2 is clearly written. However, it is not clear from this paper what is meant by propensities, and
%how propensities are related to ontic indeterminacy. It seems that a quantum framework would be more appropriate than a classical one.}\textcolor{magenta}{We can look at this problem at the end}\textcolor{blue}{i think here the question of real space and configuration space or toy model comes again}\textcolor{red!80!black}{Flavio, feel free to write sth if you know how we can possibly handle it! Did you get these type of questions in your previous papers as well?}

In \cite{del2019physics, gisin2020real, del2021relativity, gisin2021indeterminism, gisin2021synthese, del2021indeterminism, del2023prop, santo2023open, DelSantoGisin2024, del2025features}, the principle of finiteness of information density was enforced, according to which finite regions of space can contain only finite information.
Hence, in a finite universe with a finite number of events, $G(t)$, which records all information about the events and their causal structure, must itself contain only finite information. Real numbers, which typically encode infinite information, are therefore not a natural choice for representing the propensities, which we therefore assume to be rational-valued.

We will instead consider that there is a potential causal connection between the $i$-th and $j$-th events if and only if the associated propensity $p_{ij}$ is strictly positive, such that 
\[
\Omega(t) = \{ (i,j) | p_{ij} > 0,\; p_{ij} \in P(t),\; i,j \in E(t) \}.
\]
Hence, the information content of the whole graph is 
\[
I(G(t))=I(E(t), P(t)). 
\]

A fundamental property of Kolmogorov complexity is the chain rule, which holds up to an additive constant. The information can therefore be written as \begin{equation}I(G(t)) = I(E(t)) + I(P(t)|E(t)) + C_{0},
\end{equation}
where the $C_0$ is a constant independent of $E$ and $P$.
     The term $I(E)$ captures the information required to specify the set of events (nodes) constituting the graph. The conditional term $I(P| E)$ quantifies the information needed to specify the causal structure relating these events, namely, which pairs of events are connected and with what propensities, given that the events themselves are already fixed, according to the notion of event symmetry. The notion of event symmetry \cite{RidleyAdlam2023}, suggests that the physical or dynamical contribution assigned to an event does not depend on its absolute coordinate representation, but only on its intrinsic properties and its relations to other events. In this sense, coordinates are understood as observer-dependent labels, while the underlying event structure remains invariant under transformations. Different reference frames may therefore assign different coordinate values or temporal orderings to events, without altering the relational or physical content associated with the events themselves.
     This viewpoint is conceptually reconciled to Einstein’s point-coincidence argument and later relational and causal approaches to spacetime, including the causal set programme of Sorkin and Dowker\cite{Einstein1916, Einstein1923Foundation, Einstein1914, Einstein1916, Stachel1980HoleArgument, Sorkin2003, Dowker2005}. We also do not take a strong metaphysical stance, such as eternalism (block universe), growing block or presentism as an additional postulate. Instead, spacetime is treated as the kinematic background of physical events, while the central ontological focus of the framework is not spacetime itself, but the evolution of event-structure over this background.\\
     %\textcolor{gray}{On page 5, lines 38-43 of the left column, there is an intriguing remark stating
%that the events themselves are already fixed. What principle ensures that the events are fixed?} \textcolor{blue}{we can write sth similar to what we write later around event symmetry and that events are intrisically determined: We can argue and define sth around ``event symmetry'': https://arxiv.org/html/2312.13524v1/S2. In particular, an event is defined as a point in spacetime with invariant physical content, while its coordinate representation depends on the chosen reference frame. Under Lorentz transformations, the coordinates assigned to events change and the temporal ordering of events may vary, but the underlying event structure remains invariant.}\textcolor{gray}{Could this apply to the "connections" regarding both their direction and magnitude?.}\textcolor{blue}{I think the connections are definitely not fixed. The events are located at fixed points (which has refere pt dependent coordinates) because of their inrinsic existence. The causal arrows connect potential caused events, one of which actualizes in time.}
Additionally, in our causal graph, we can also separate the information content of the graph $G(t)$ into a part relative to the past, and one relative to the future, and the following relations should hold: \\

\begin{enumerate}[label=\Roman*), leftmargin=1.5em]
    \item
    $I(G(t))_{\text{past}} = I(E(t))_{\text{past}}  + I(P(t)| E(t))_{\text{past}}+ C_P$\\
    
    and\\
    
    $I(G(t))_{\text{fut}} = I(E(t))_{\text{fut}} +I(P(t)| E(t))_{\text{fut}}+ C_F$\\

    where $C_P$ is the constant in the past causal graph, and $C_F$ is the constant in the future (represented as fut in the equation) causal graph. Note that the constants depend on the embedding scheme used for calculating the Kolmogorov complexity. \\

   % \textcolor{gray}{Regarding the expressions at the end of the left column on page 5: how can the constants CP and CF be calculated? Are they related to each other?}\textcolor{blue}{Their values therefore depend on the specific modeling choices, depending on the embedding scheme used to calculate the Kolomogorov complexity}\textcolor{red!80!black}{Okay, as I can see this is already mentioned, not sure what they want} \textcolor{gray}{Additionally, since the second term of each deterministic expression pertains to transitions, is the geometric time related to the creativity of the first term?}\textcolor{red!80!black}{Dont understand what it means}
    
    \item 
    $I(P(t)| E(t))_{\text{past}} < I(P(t)| E(t))_{\text{future}}$
      \end{enumerate}
      
%\textcolor{gray}{In the right column on page 5 and the subsequent discussion, how do propensities “collapse”? Is this behavior analogous to quantum mechanics (see also question f)?} \textcolor{blue}{The term is used in a more general and operational sense to describe the transition from a set of multiple admissible potential continuations to a single realized outcome.It is similar to Quantum mechanics in the sense, the collapse also doesn't occur in real space, but some sort of configuration space}\textcolor{red!80!black}{General question, can potentialities be amplitudes, or just classical or "quasi" classical weights? now I am thinking the superposition-entanglement question was based on that} 
%\red{Yes, write this in the reply. Reference my paper "potentiality realism" and the redentece of M. Suarez that I added. Buth Popper's original proposal for propensities and Mauricio claim that one can do QM with propensities if they are amplitudes and can interfere.}\textcolor{gray}{What is the source of information generation for the future?}\textcolor{blue}{something around: expected information (not accessed) through time-symmetric dynamic laws?}

since the past holds a record of actualized events, where propensities take extremal values, while the future encodes open alternatives that require richer informational content (i.e., the information needed to encode a general rational number)\footnote{Because information storage is finite, the earliest records are gradually erased, as we move forward in time, leaving our knowledge of the remote past effectively indeterminate and open again. The indeterminacy of the “open past” differs from the indeterminacy in the future because it never evolves into actualization. Hence, the past still reflects an essential structural asymmetry: the recent past is remembered while the future is partly determined, therefore ``open".} %\textcolor{olive}{Here, the term ``collapse" is used in a more general and operational sense to describe the transition from a set of multiple admissible potential continuations to a single realized outcome. It is similar to Quantum mechanics in the sense, the collapse also doesn't occur in real space, but some sort of configuration space.}
.

%\textcolor{gray}{[however, if the real limit on classical information density is just a symptom of the limits of classical physics,
%Then the conclusion of the information argument is not that classical physics harbours indeterminacy, but
%rather that classical physics is not valid in the infinitesimal limit; and a new mechanics is needed. In this
%case: does the no-go argument of the paper still hold?]}
%\textcolor{magenta}{We can look at this problem at the end}
%\red{we have already replied in detail to this problem. Look at what I added in red in the reply}

 %Also, it can be compared to Bob Coecke's notion of terminality or John's tracing out conditions in the process matrix formalism with fixed background)}
%\textcolor{blue}{We can maybe also write sth like within this setting, we do not take a strong metaphysical stance such as eternalism (block universe), growing block or presentism as an additional postulate. Instead, spacetime is treated as the kinematic background of physical events, while the central ontological focus of the framework is not spacetime itself, but the evolution of event-structure over this background.}

\section{No-go theorem for Superluminal observers}

%\textcolor{gray}{• Presumably unrealized events do not feature in the manifold? (If we do not know what occurs at
%E(t,x,y,z) then we cannot allow that event to in any sense ‘exist’.)
%• If we allow two categories of time, then events may exist in one, but not the other.
%• If certain events have not actualized, then in a standard diagram of the spacetime manifold these
%unrealized events do not exist. I therefore question the validity of Fig.1, which depicts a standard
%manifold, but with missing parts.}
%\red{In our formalism, the manifold encompasses all potential events. Some of these will never actualize. This is why we have a multigraph, where the different types of edges distinguish exactly this. The nodes, hoever, represent all possible events.}

Actualization of propensities in a space-time with finite information density reveals a time-asymmetric structure where the past is a record of recently actualized events, and the future consists of potential alternatives. Under a SpLT \footnote{This argument can be extended to any non-order preserving SpT where the information structure after transformation does not align parallel to the space-time structure.}, the coordinates are swapped (i.e., position becomes time and vice versa in 1+1 dimensions) when the reference frame is boosted in the limit of infinite speed. From figure \ref{f3.3}, we see that the past and future densities of information are, in fact, symmetric over the $t'$ axis under a SpT, a simple coordinate transformtion relocating the observer to a superluminal reference frame.

\begin{figure}[]
		\centering
		\includegraphics[width= 0.48 \textwidth]{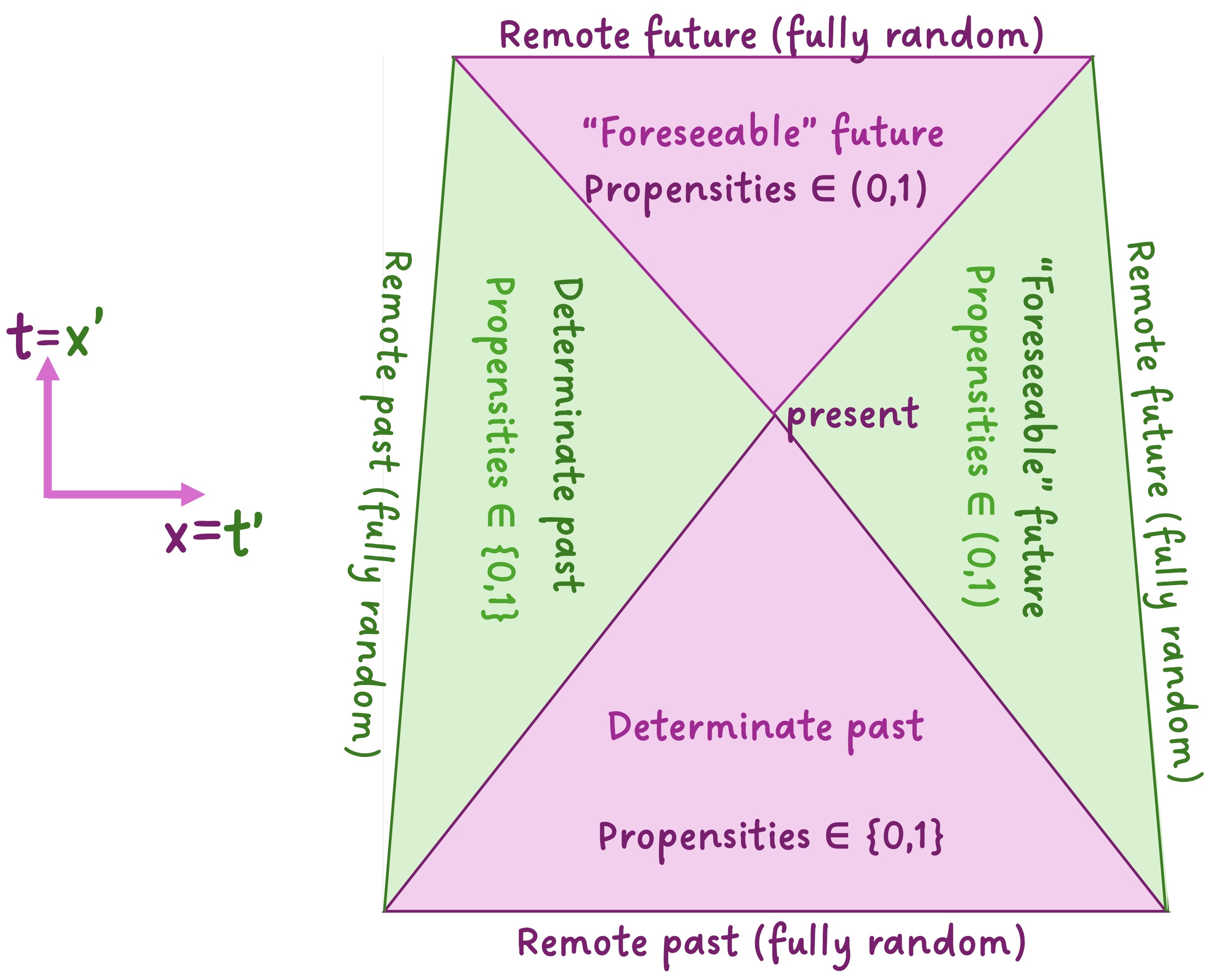}	\caption{Schematic representation of the informational structure associated with the light cone. The coordinate in purple denotes the reference frame of a subluminal observer, and the one in green represents a superluminal observer.  The diagram illustrates how regions identified as past and future are organized relative to the light cone and how they may be related under the allowed set of transformations. The figure is conceptual and is intended to highlight how the distinction between past and future, and the corresponding informational structure, depend on the underlying assumptions of the framework.}
		\label{f3.3}        
	\end{figure}

%\textcolor{gray}{Figure 2 on page 5 requires a clearer description.}

This unusual symmetry raises questions about the status of coordinates of space-time and information-causality under SpTs. As with other interpretational claims, a toy theory involving superluminal extensions require careful consistency analysis, for which no-go theorems can provide useful constraints. Every operational phenomenology implicitly entails an interpretational stance, one that can ultimately be supported or ruled out by such theorems.

To better understand the foundations and implications of such transformations in the form of a no-go theorem, consider the following assumptions:\\

A0. \textbf{Existence of non-order-preserving SpTs}\\

%\textcolor{gray}{The A0 assumption on page 6 requires clearer writing and explanation.}
A transformation with relative velocity $|v| > c$ that maps events which are labelled as timelike-separated events to events which are labelled as spacelike-separated events in a different reference frame, and vice versa, such that:

\[
\exists \hspace{0.1cm} E_1, E_2 \in E : E_1 \prec E_2 \text{ but } S_v(E_1) \prec S_v(E_2) \text{ is false}.
\]

where $S$ is the SpT and $E_1$, and $E_2$ are events in the set $E$. For example, the anti-symmetric Lorentz extension in Eq \eqref{eq: sup} with $v > c$ where $x \rightarrow t''$ and $t \rightarrow x''$  at $v \to \infty$ in dimensions 1 + 1.\\

\begin{figure}[htp]
		\centering
		\includegraphics[width= 0.4 \textwidth]{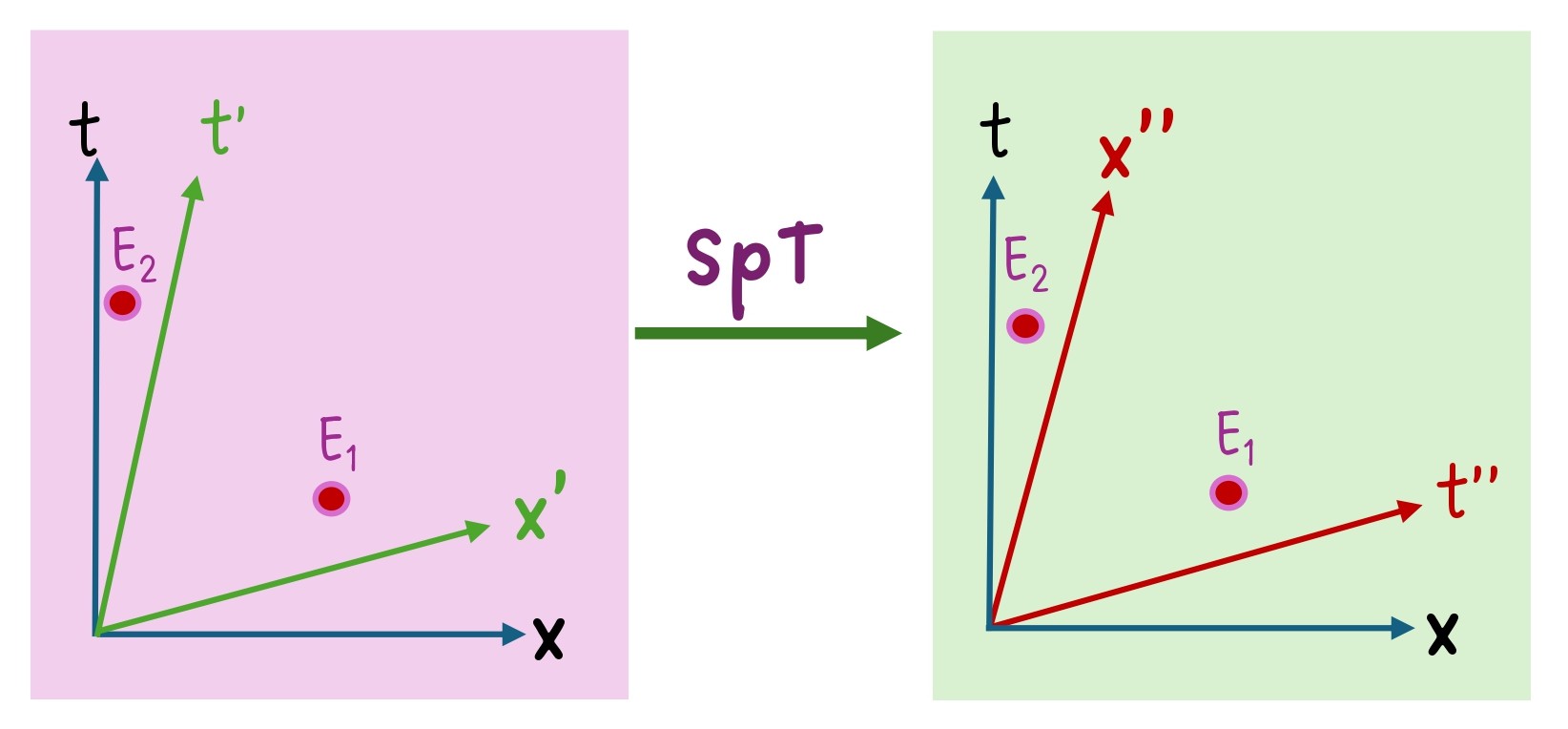}	\caption{Illustration of event ordering under subluminal and superluminal transformations. In the left panel, the green coordinate system $(x',t')$ is related to the blue frame $(x,t)$ by a subluminal transformation. Since both frames are subluminal with respect to one another, they preserve the temporal ordering of the events, with $E_1$ occurring before $E_2$. On the right, the red coordinate system $(x'',t'')$ corresponds to a superluminal transformation relative to the original blue frame. Under this transformation, the temporal ordering of the events is no longer preserved, such that $E_2$ precedes $E_1$ in the $(x'',t'')$ description. The figure schematically illustrates how extended superluminal transformations can alter the observer-dependent ordering of events while retaining the same underlying event structure.
}
		\label{orderchange}        
	\end{figure}
%\textcolor{blue}{We additionally assume a form of event symmetry \cite{RidleyAdlam2023}, according to which the physical or dynamical contribution assigned to an event does not depend on its absolute coordinate representation, but only on its intrinsic properties and its relations to other events. In this sense, coordinates are understood as observer-dependent labels, while the underlying event structure remains invariant under transformations. Different reference frames may therefore assign different coordinate values or temporal orderings to events, without altering the relational or physical content associated with the events themselves.}\textcolor{magenta}{This viewpoint is conceptually related to Einstein’s point-coincidence argument and later relational and causal approaches to spacetime, including the causal set programme of Sorkin and Dowker.\cite{Einstein1916, Einstein1923Foundation, Einstein1914, Einstein1916, Stachel1980HoleArgument, Sorkin2003, Dowker2005}}
Under superluminal transformations, the coordinate representation of events may change, but the events themselves are treated as invariant relational structures (according to event symmetry).\\ %\textcolor{blue}{Importantly, their existence as kinematical mappings does not, by itself, imply superluminal signalling. Signalling is defined at the operational level in terms of controllable preparation, evolution, and measurement of physical systems, which remain constrained by the light-cone structure of Minkowski spacetime.}\\

\blk

A1. \textbf{Principle of finite information}\\

The concept of a finite universe, together with the principle of finite information density, posits that the universe consists of a finite number of physical resources confined within a finite spatial volume. Hence, at a given instant, the total amount of information it contains is likewise finite, independent of the choice of the reference frame. Equivalently, any spatially bounded region of space on each temporal hypersurface can encode only a finite quantity of information \cite{del2019physics, gisin2020real, del2021indeterminism, PhysRevD.23.287, thooft1993dimensional, susskind1995world, bousso2002holographic, srednicki1993entropy, harlow2016jerusalem, bekenstein2003information}.\footnote{This idea is also supported by the holographic principle, which implies the Bekenstein bound: any finite region of space can store only a limited amount of information~\cite{PhysRevD.23.287}. Since storing information requires energy, and excessive energy density would cause gravitational collapse, the total information content of the universe must be finite.} Analytically, a finite spacetime region contains at most $N < \infty$ bits of information such that:
\begin{equation}
    I(G(t)) = I(E(t)) + I(P(t) | E(t)) + C_{0}  \leq N.
\end{equation}

where $C_0$ is constant.\\

%Also, the total information content of the universe \textcolor{blue}{at a given instant} is constant and does not depend on the choice of reference frame. \textcolor{blue}{Here, ‘total information content’ refers to the finite information associated with the universe at a given instant t. Thus, the principle of finite information is understood as a bounded informational description on each temporal hypersurface.}\\

%\textcolor{gray}{[Q3] What spacetime ontology do the authors assume or adopt? A form of growing block? The assumptions A2
%(time symmetry of total information) and A3 (past memory) seem to be very dependent upon this?}\\

%\textcolor{blue}{I guess our ontology is kind of dynamically expanding and tranforming state space, so yes growing block? Or we do not take a stance? can it be presentism, since past is different in different reference frames?}\\
A2. \textbf{Time symmetry of total information}\\

%The principle of homogeneity of time (or the symmetry of time translation) asserts that the fundamental laws of physics are invariant under translation in time. Formally, the informational structure is invariant under \( t \mapsto t + \tau \), where \( t \) is the time coordinate and \( \tau \in \mathbb{R} \) represents a temporal translation.

We assume that the maximal storable information in the universe is expected to be approximately symmetric in time with respect to the present in all individual inertial reference frames. So, if there are in total N bits of storable information in the universe, about N/2 bits are used to record information about the past, that is, $I(G(t))_{\text{past}}$, and the other N/2 are recording $I(G(t))_{\text{future}}$, making $
I(G(t))_{\text{past}} \approx I(G(t))_{\text{future}}$.\\ 

This assumption was introduced in Ref. \cite{santo2023open}, and motivated by homogeneity of time.  It is true that there is no general physical principle that guarantees an equal split of the existing information between past and future. However, there is no a priori reason for the universe to allocate its total storable information in a strongly asymmetric way between past and future records.\footnote{In standard frameworks of fundamental physics, dynamical descriptions employ time-symmetric equations of motion, that do not, by themselves, distinguish a preferred temporal direction. This structural feature suggests that, at the level of these laws alone, there is no built-in asymmetry in how physical histories are generated or encoded, independently of any particular theoretical interpretation or ontological commitment.}  A2 is not a claim about equality of \textit{accessibe information} content, but about the \textit{existing} information that encodes all the events. The time-symmetry of the underlying dynamical laws and state descriptions therefore suggests that there should be no fundamental difference between the structure of the past and of the future.\\% In standard frameworks of classical physics and quantum physics, microscopic laws are (approximately) time-reversal symmetric and no fundamental asymmetry in the laws themselves that privileges past over future. The asymmetry in \textit{accessible information} arises from macroscopic boundary conditions (e.g. low-entropy past due to second law of thermodynamics), not from the underlying dynamical structure. A2 is thus a weaker assumption, at the level of underlying laws.}\\

A3. \textbf{The past holds memory} \\ 

In any inertial reference frame, the past stores determinate events (records of outcomes that have already actualized), while the future contains only potential events represented by nonextremal propensities. The past records actualized events, such that
\[
P(t)_{\text{past}} \in \{0,1\},
\]
while future propensities can take any rational value in $[0,1]$, encoding potentialities:
\[
P(t)_{\text{future}} \in \mathbb{Q} \cap [0,1].
\]
That is, it records the events that actualized recently, reducing their propensities to  0 for the events that did not take place with certainty and 1 for the events that occurred with certainty. This is at the level of \textit{accessible information}, since past is already accessed/ actualized. Thus, the amount of bits needed to encode rational propensities relative to the future is, in general, greater than the past propensities that each contain only one bit of information since they are already actualized making $I(P(t)| E(t))_{\text{past}} < I(P(t)| E(t))_{\text{future}}$.\\
%\footnote{\blue{We expect roughly the same number of events and edges in the past and in the future, that is, $I(E(t))_{past} + I(\Omega(t))_{past}= I(E(t))_{fut} + I(\Omega(t))_{fut}$}}.\\ 

A4. \textbf{Time has a specific role in every reference frame}\\

Time has a specific role: along it, the causal structure is aligned, and the envelope of the light cone within which actualization can occur is defined. Under any reference frame transformation, it should remain the parameter along which deterministic equations evolve, and, more importantly, the parameter relative to which causal relations are defined.
%\red{We do not distinguish any more creative and geometric time. Double check that this distinction is left not anywhere.}

%\textcolor{gray}{Similarly, A4 seems too vague and at the same time too strong. It is not clear to me what exactly counts as a “physical” reference frame transformation here, and the requirement that time must always retain this specific structural role already seems to exclude part of the conceptual space that the theorem is supposed to probe, especially in light of superluminal Lorentz transformations for which space and time become interchanged in an appropriate limit.}

While different observers may disagree on the decomposition of spacetime into spatial and temporal components. In any chosen inertial reference frame, the decomposition of spacetime into temporal and spatial directions induces a corresponding representation of the light-cone structure, with timelike directions with respect to space-time points within the light cone, lying within the cone defined relative to that frame’s time axis. Different observers adopt different decompositions, but all such representations are related by subluminal or superluminal Lorentz transformations.

\newtheorem*{theorem*}{Theorem}
\begin{theorem*}
The conjunction of all assumptions A0.- A4. cannot hold simultaneously. 
\end{theorem*}

\begin{proof}

%\textcolor{blue}{according to our re-clarified story, I (G) is the total information mased to microscopic laws and needn't be accessible, I(P) however is accessible. Maybe we need to reclarify in the proof again.}\\

Let us assume A0--A4 all hold. From A1 (finite information), the total information, we have
\begin{equation}
I(G(t)) = I(E(t)) + I(P(t) | E(t)) + C_{0}  \leq N
\label{eq:total}
\end{equation}
and from A2 (time symmetry of total information), we expect
\begin{equation}
I(G(t))_{\mathrm{past}} = I(G(t))_{\mathrm{future}}, 
\label{eq:balance}
\end{equation}
Thus, each side of \eqref{eq:balance} is finite and bounded by $N/2$.\\

From A3 (past holds memory), propensities in the past are extremal and require strictly fewer bits to encode them than the non-extremal ones. Thus,
\begin{equation}
\Delta_P := I(P(t)|(E(t))_{\mathrm{future}} - I(P(t)|(E(t))_{\mathrm{past}} > 0.
\label{eq:propensitygap}
\end{equation} 
%where $\Delta_P$ \footnote{\blue{We can expect $\Delta_P$ should be a frame-independent quantity since the propensity of each individual event's occuring should not be dependent on the frame of reference. However, it is not a necessary assumption for our proof, since this is a proof of contradiction specified by a particular example.}} should be a strictly positive quantity. %since the indeterminacy of an event \blue{before it actualizes should not depend on the frame of reference.}

%Assuming that the contributions of events are symmetric under time translation, the equality in \eqref{eq:balance} can only hold if the larger future propensity information is compensated elsewhere (a bigger area for the information content assumes a larger number of nodes in the past in Fig.\ref{f3.4}). \textcolor{blue}{We note, in a standard subluminal reference frame, the number of events in the past should be more than in the future to account for the past that the indivudual propensity in the past is more compressible (less information), but the total expected informational content in the past and the future is same.}
Consequently, if the total information in past and future remains approximately balanced as required by A2, the event contribution must compensate for this asymmetry, making : $I(E(t))_{\mathrm{past}} > I(E(t))_{\mathrm{future}}$.  Thus, the left hand side of \ref{f3.4} is the most general description. Note, we can always have at least one scenario with comparable $C_P$ and $C_F$ to get this relation because the constants are subject to the embedding scheme chosen by the observer. Because time has a role (assumption A4), each reference frame still defines its own temporal direction and associated light-cone structure.

Now consider assumption A0, i.e., there exists a SpT, let us say $S$, that does not preserve order. 
%\textcolor{magenta}{maybe we need to re-write in a better way because now, we have the total information as the total information as a given instant, and the transformation probably takes atleast delta t time}
 %\blue{Total information is invariant(\textit{additional assumption})}:

To remain consistent with the principle of relativity and the absence of any preferred inertial reference frame, the total information content must be invariant under changes of reference frame, implying that: 
\begin{equation}
I(G) = I'(G),
\label{eq:invariance}
\end{equation}

where the primed information count is the information count in the new superluminal reference frame. The finite total number of events is fixed according to event-symmetry and should be distributed in a frame-independent manner. %\textcolor{blue}{We understand events, indendent of its abolute coordinates where the physical and dynamic contribution assigned to an event is only intrinsic and relational to other events.}
We then achieve the distribution of events under SpTs as on the right in Fig. \ref{f3.4}

\begin{figure}[h]
    \centering
    \includegraphics[width=1.0\columnwidth]{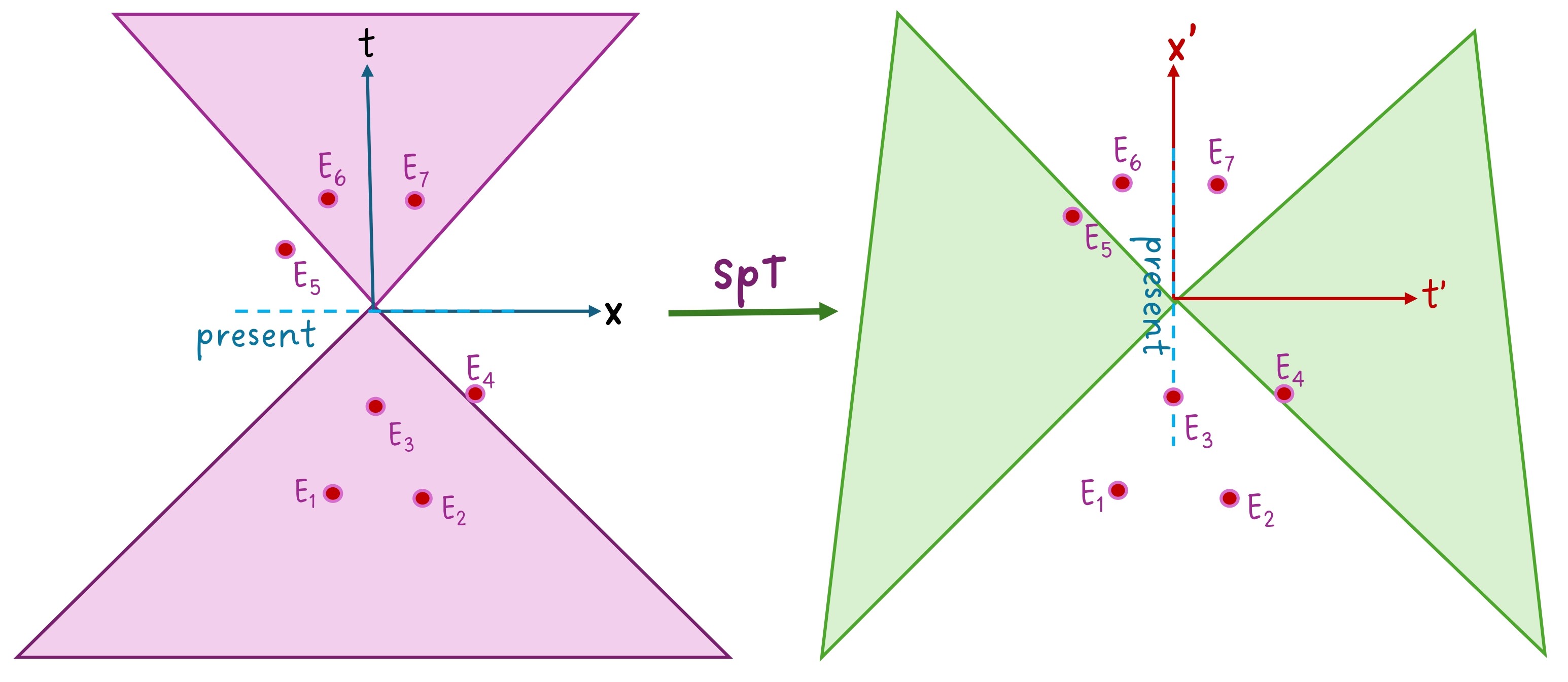}
    \caption{Schematic representation of the evolving causal structure underlying the framework. The diagram depicts a set of events connected through causal relations constrained by the light-cone structure, with temporal ordering defined relative to an observer-dependent decomposition of spacetime. On the left we have the total information structure in pink with events $E_1$ to $E_7$ which undergo an SpT transformation (A0) on the right, where the lightcone is given in the direction of time (A4).} 
    \label{f3.4}
\end{figure}
%\textcolor{gray}{Figure 3 on page 7 also needs a better description.}

Under A0, A3, and A4, we see that $I'(E(t))_{\mathrm{past}} = I'(E(t))_{\mathrm{future}}$  in the superluminally boosted reference frame, thereby violating A2 and \eqref{eq:balance} (making $I'(G(t))_{\mathrm{past}} \neq I'(G(t))_{\mathrm{future}}$) or A1 and \eqref{eq:total} (making $I'(G(t)) = I'(G(t))_{\mathrm{past}} + I'(G(t))_{\mathrm{future}} \leq \infty$). %In the transformed frame, we denoted the past/future contents by $I'(G(t))_{\mathrm{past}}$ and $I'(G(t))_{\mathrm{future}}$. 
Hence, A0--A4 are mutually inconsistent.

\end{proof}

\section{Relaxing the Assumptions}
In this section, we will discuss the consequences of violating each assumption.\\

$\lnot$ A0. $\rightarrow$ {\textbf{Non-existence of SpTs}\\ %\m{reated to part before?}

    %\red{Here I leave it to you. Don't overcomplicate things. Simply tell that one can abandon SPTs altogether as there is no physical evidence that they exist. However, contrarily to what it is largely believed, they do not in principle violate causality and it is therefore interesting to see their limits of compatibility with other physical assumption as a conceptual tool.}

One assumption in our analysis is the possible existence of superluminal transformations (SpTs) as kinematical extensions of Lorentz transformations. This assumption can be relaxed by excluding SpTs as physically meaningful structures. That is, we can remove them both as mathematical objects and as candidates for physical symmetries, a priori, in the absence of empirical evidence. If SpTs are excluded, the framework reduces to standard relativistic physics, in which the assumptions A1–A4 are satisfied as in Ref. \cite{santo2023open}.\\

To note, the existence of SpTs in a purely mathematical sense, as just kinematical mappings does not, by itself, imply superluminal signalling. Signalling is defined at the operational level in terms of controllable preparation, evolution, and measurement of physical systems, which remain constrained by the light-cone structure of Minkowski spacetime. In this weaker sense, they serve as transformations at the level of representation rather than as elements of physical symmetry or dynamics. Thus, allowing mathematical SpTs modifies only the kinematical description of spacetime, while leaving the operational causal structure unchanged. No superluminal signalling arises at the level of physical processes.\\

$\lnot$ A1. $\rightarrow$ \textbf{Violation of the principle of finite information}\\

A world with SpTs (A0) can exist at the expense of violating the principle of finite information $(N \rightarrow \infty)$. In such a setting, both the past and the future would contain infinite information, rendering any comparison of their respective contents meaningless (A2 and A3). %One could still impose $I(G(t))_{\text{past}} \approx I(G(t))_{\text{future}}$ (A2.), $I(P(t))_{past} < I (P(t))_{future}$ (A3.) in the presence of SpTs (A0.) with actualization in the temporal direction specified by A0.(A4.)
A violation of finite information also allows us to encode infinite information into mathematical real numbers, an assumption of classical deterministic theory. But, in a world with finite information, physical quantities are determined only to a certain level of precision (like in FIQs \cite{del2019physics}), which leads to classical indeterminism. %One way to put it is this: if we allow superluminal transformations to exist, then real numbers can be used to define physical quantities specified with infinite precision.
 Interestingly, this observation reveals an intriguing question about the relation between two sources of indeterminism: the claimed indeterminism associated with superluminal observers \cite{dragan2020quantum}, and the indeterminism arising from the use of FIQs in classical theories \cite{del2019physics, gisin2020real}. 

In the regime of subluminal transformations, the corresponding energies are always positive, but in the regime of superluminal transformations, we encounter unavoidable negative energies, \cite{Schwartz2018, sen2026superluminal}. Superluminal particles (tachyons), for example, can have positive energies in one reference frame, and can be transformed even with a subluminal Lorentz transformation to a new reference frame where they have negative energies. The problem with the negative energies can be solved by simply reinterpreting the particle's energy and momentum with
\begin{equation}\label{reinterpet}
    (E, p) \rightarrow (-E, -p),
\end{equation}
extending the solution of \cite{Schwartz2018} to the case of superluminal boosts. This is analogous to reinterpreting a negative energy ``incoming'' particle as a positive energy ``outgoing'' particle (or vice versa)\cite{Schwartz2018, sen2026superluminal}. 

But, for finding the negative energy states, we need to compare the real-numbered velocities to the E-p relations \footnote{See Sec.~IIIA of~\cite{sen2026superluminal}, where the precise value of velocity is required to compare against $\sqrt{\frac{m^2}{p^2}+1}$ and determine the sign of energy. This is unnecessary for purely subluminal transformations, since physical subluminal particles can always have positive energies.}, which are again dependent on real-numbered quantities- mass and velocity. This requires infinite precision and therefore infinite information, closely mirroring the ontological structure of orthodox classical theory based on real numbers.\footnote{%\red{Too long and complex footnote. I dont' see a big value in this. Please simplify it greatly or even get rid of this.} 
A theory containing SpTs together with infinite information would also possess all the supplementary variables required to determine the evolution uniquely. With infinite precision available for specifying and comparing velocities at any time, one can always formulate deterministic dynamical laws yielding real-valued quantities at every instant. This presumed determinacy can be extended to Quantum Reference Frame transformations; in such a theory, reference frames may be in superpositions of any combination of subluminal or superluminal velocities \cite{sen2026superluminal}}.

$\lnot$ A2. $\rightarrow$ \textbf{Violation of time symmetry of total information}\\

A possible explanation for the distribution (as shown in figure \ref{f3.3} for SpLT) of information under SpT would correspond to a violation of the time symmetry of total information. 
This would mean that the informational record, in which the universe records the information about events and their causal relations is not symmetrically around the present moment. Hence $I (G(t))_{\text{past}} \neq I (G(t))_{\text{future}}$. The total information may still remain finite (A3), in a manner consistent with the existence of superluminal transformations (A0) and the arrow of absolute time (A4). \\

%However, we consider this assumption so naturally justified that \as{we would be quite} reluctant to relax it.\blk\\  

%\blue{maybe we don't take strong sides but leave it on the reader with what we want to say being one possible interpretation? What do you think?}\\

$\lnot$ A3. $\rightarrow$ \textbf{Immediate past holds no memories and is random like the future} \\

%This is a violation of finite information.

%\n{This assumption if considered consistent with the homogeneity of time (assumption 2); that is the total information is symmetrically divided between past and future is inconsistent with the violation of finite information (assumption 3) in the presence of Superluminal transformation. That is, in the presence of infinite information}\\

SpLT \cite{parker1969}, an example of SpT (A0), suggests that under superboosts at infinite speeds, the space and time coordinates interchange. Hence, a possible interpretation of the assumption of temporal symmetry of expected information (A2) in the coordinates of such a SpLT would be $I(P(t)| E(t))_{\text{past}} = I (P(t)| E(t))_{\text{future}}$. That is, the content of information through the propensities is the same both in the past and future. This implies that the propensities in the past and the future are the same, and the past does not retain any memory. Hence, propensities in the past, rather than being extremal values (0 or 1), should become nonextremal, rational-valued weights reflecting a lack of determinate outcomes. Not only is the memory of the distant past erased, but even the immediate past ceases to retain a record, rendering it structurally similar to the future. %With an equally distributed total information in the past and future, we could also assume an equal number of events and edges $I(E(t))_{past} + I(\Omega(t))_{past}= I(E(t))_{future} + I(\Omega(t))_{future}$\textcolor{magenta}{?}.

%This is consistent with the direction of events or time (A4) of \ref{f3.2}. 
Abandoning a past with memory aligns with the notion that information is actualized only within light cones under a violation of (A4), considering the structure of the lightcone doesnt change. For example, if one were to travel in a spacelike trajectory or undergo a SpT, the fundamental variables would go from indeterminate $\rightarrow$ determinate (at present) $\rightarrow$ indeterminate, which is the other way of saying that the past and future in this case are equally indeterminate.  %\textcolor{magenta}{This perspective also aligns with the concept of finite information (A3), where finite information can respect temporal homogeneity while being approximately equally distributed in the past, present and future.}\\

In a scenario where the asymmetry between past records and future possibilities becomes weakened, the informational distinction between memory and prediction may cease to remain fundamentally meaningful. One might be tempted to notice some conceptual resemblance to Maxwell’s Demon where the physical role of memory, record keeping, and accessibility of information becomes central to the emergence of thermodynamic irreversibility. If superluminal transformations blur the distinction between past and future informational structures, then the notion of entropy increase associated with stable memory records may also require reinterpretation. While we do not formulate a thermodynamic model in the present work, these considerations suggest a possible connection between superluminal informational structures, entropy, and the arrow of time that may be interesting.\\

$\lnot$ A4. $\rightarrow$ \textbf{Time has no specific physical role in superluminal coordinates}\\ 

Abandoning A4 would imply that time, as a coordinate, has no physical role. It would no longer serve as a direction of information flow but merely as a label in the space-time manifold. 
 This would also imply that the ‘space’ and ‘time’ coordinates are not qualitatively distinct, thus challenging the physical meaning of the arrow of time---or the arrow of actualization.

Studies separating space-time from the information-theoretic structure have been conducted in \cite{vilasini2022general, Vilasini2024, Vilasini2024a}, providing a formalism to analyze causation, causal loops, and signalling depending on the space-time embedding of a given causal structure. We are in the process of formalizing this notion for space-time variables or events embedded in a causal structure that transforms under SpT, which will appear in forthcoming work. However, for this paper, we do not characterize the information-theoretic causal structure based on the geometry of space-time, but talk about the total informational structure with respect to SpLTs. %From this perspective, actualization in creative time and the flow of informational structure should be inherently assumed to follow the temporal direction.
 %\textcolor{orange}{write this better}\\

%5.\as \textbf{actualization remains untouched under Superluminal transformation}
%(relativistic transformations of creative time)\\

 %Flavio, is this not a spin off of 0.? 

%Added something at 0. \\

\section{Discussion}

Our result represents a deliberate theoretical claim, akin to many foundational proposals, (such as, for example, \textit{P-R boxes} \cite{popescu1994quantum} in hypothetical generalized probabilistic theories) as it involves superluminally boosted transformations or objects. Yet such frameworks provide a way to test the internal consistency of our foundational physical principles when pushed beyond their conventional domains of applicability. The study of superluminal transformations, in particular, can provide a lens through which foundational tensions regarding (in)determinism, information, and the arrow of time can be reexamined. In this way, they confront us with a choice between formal symmetry and mathematical completeness on the one hand, and physical plausibility grounded in finite information and a determinate past on the other.

%\textcolor{olive}{Why should one investigate superluminal transformations at all, particularly in the absence of their experimental evidence? One important reason is that }\\

This motivation also highlights two distinct ways of thinking about causality in spacetime. In the first approach, events and their causal relations are statistical and treated as primitive, while spacetime geometry is emergent- for example, the causal set theory program developed by Sorkin, Dowker, and collaborators \cite{Sorkin2003}. A second and more traditional viewpoint, commonly adopted in general relativity, begins with spacetime points and worldlines as primitive geometric entities, and events defined operationally on them. Our present work philosophically and structurally promotes the former approach.

An important aspect of the present framework concerns the role of finite information. Our assumptions are formulated at a structural and theory-independent level, without committing directly to quantum theory. Nevertheless, one might question that there are suggestive conceptual parallels with quantum limitations such as the energy–time uncertainty relation. Both the finite-information viewpoint and energy–time uncertainty reflect a common structural feature: that physical processes limit the rate at which distinguishable states can be resolved or generated within a finite time interval. While we do not derive quantum uncertainty relations from our framework, since it is beyond the scope and interest of this paper, these similarities may point toward deeper insights about physical theories with finite information.

Our results suggest that the indeterminacy associated with superluminal transformations may differ fundamentally from quantum indeterminacy. Within a finite-information framework, apparent indeterminism can arise from informational constraints and limited accessibility to underlying variables, rather than from incompatible observables or intrinsically probabilistic dynamical laws \cite{del2023prop}. In standard quantum theory, indeterminacy is closely tied to the measurement postulate, where measurement outcomes are not simultaneously definable for incompatible observables, and the act of measurement plays a fundamental role in actualizing physical outcomes. By contrast, the indeterminacy discussed here emerges at a more structural and informational level, independent of the quantum measurement formalism itself.

The notion of indeterminacy introduced by Dragan and Ekert proposes that SpLT leads to fundamental indeterminism. In such a world, no such thing as a truly classical system exists, and classicality becomes an approximation; a perspective that also underlies the broader programme of Quantum Reference Frames. In related work, we investigated the implications of superluminal quantum reference frames for physical principles governing energy, entropy, and Bell inequalities \cite{sen2026superluminal}.\\

%\textcolor{gray}{6.\textcolor{red!80!black}{General question, can potentialities be amplitudes, or just classical or "quasi" classical weights? now I am thinking the superposition-entanglement question was based on that} 
%\red{Yes, write this in the reply. Reference my paper "potentiality realism" and the redentece of M. Suarez that I added. Buth Popper's original proposal for propensities and Mauricio claim that one can do QM with propensities if they are amplitudes and can interfere.}}

In the present work, we adopt a minimal and more general notion of propensities with objective causal weights associated with possible events, following the tradition initiated by Popper’s propensity interpretation \cite{popper1959propensity} and further developed in, e.g.,  \cite{ballentine2016propensity, suarez2010probabilities, gisinpropold, del2019physics, del2023prop}. In this formulation, propensities function as causal weights assigned to alternative potential outcomes and are not assumed to exhibit inherently quantum-mechanical structure or interference effects. However, as noted by Suárez \cite{suarez2010probabilities} and Weinert \cite{weinert2025propensity}, such a propensity-based ontology can in principle be extended to reproduce quantum phenomena if propensities are promoted beyond classical probabilities to complex-valued amplitudes capable of interference. Consequently, questions of superposition, entanglement, or collapse are not excluded by the framework, but they are not just required for the no-go theorem developed here. Note that, since propensities are operationally accessed via relative frequencies in the weak law of large numbers, they can be represented effectively by rational numbers, thereby naturally satisfying a finite-information principle (see \cite{del2023prop} for discussion).

\section{Conclusion}

%\textcolor{gray}{Comment: arguments against SpT are given, (not just against the indeterminacy issue).}
Orthodox classical physics is traditionally viewed as deterministic; however, this determinism rests on the assumption that initial conditions are specified using real numbers—mathematical objects that contain infinite information. As a result, if we adopt a more realistic notion in which only finite information can be encoded in any physical system, classical trajectories become non-unique, and the theory becomes inherently indeterministic. In this framework, indeterminism in classical theory does not arise from dynamical stochastic laws but from the fundamental kinematical indeterminacy~\cite{del2019physics, gisin2020real, del2021relativity, gisin2021indeterminism, gisin2021synthese, del2021indeterminism, del2023prop, santo2023open, DelSantoGisin2024, del2025features}.
%\textcolor{teal}{\textbf{Will change most of the discussion}}\\
In another context, the authors of~\cite{dragan2020quantum} argue that extending Lorentz symmetry to superluminal regimes explains indeterminism in a way akin to quantum theory.\\ %While such transformations can blur the distinction between past, present, and future, one interpretation of our no-go theorem is that finite-information constraints are fundamentally incompatible with superluminal transformations. Consequently, the indeterminism discussed in~\cite{dragan2020quantum} cannot be of the same fundamental nature as indeterminacy arising from finite-information limitations, and---under certain assumptions---there is no possibility of reconciling the two.

%not This ambiguity, when combined with the principle of finite information density, necessitates a reevaluation: if superluminal observers exist, then the information accessible to them must be infinite and, possibly, determinate. In this work, we see, the so-called “quantum” indeterminacy of superluminal observers, as proposed by~\cite{dragan2020quantum}, can be instead seen as a fundamentally deterministic theory within a real-valued universe endowed with infinite information—akin to that assumed in orthodox classical theories.

%In this work, we argue that the indeterminism introduced by superluminal reference frames is not an exotic feature that points to quantum randomness, but a classical phenomenon rooted in the representation of physical quantities with real numbers in orthodox classical deterministic theories.This reframes such indeterminism not as a signature of quantum randomness, but as a consequence of classical models that does not unrealistically assume infinite precision.}
We presented a no-go theorem demonstrating that the coexistence of superluminal frames with finite information, time-symmetry of total information, memory of the past, and a consistent arrow of time cannot all be maintained simultaneously. An intuitive argument for the no-go theorem may also be given. Past events are fully determined while future events remain potential. Superluminal transformations, however, can change the temporal ordering of events, mapping events that are not yet fully determined in the future of one inertial reference frame to past events in another. Thus, in this scenario, it is impossible to maintain both a determinate past in superluminal frames and a finite-information framework with a meaningful direction of time.\\

Rather than privileging a specific interpretation, our result primarily highlights the mutual tension among the fundamental assumptions of a physical theory, allowing for more than one interpretation that can accommodate different theoretical commitments.
However, one interesting interpretation of our no-go theorem can suggest that a world admitting superluminal transformations may point towards orthodox classical determinist ontology as a candidate interpretation. In other words, the existence of superluminal observers is permissible only in a universe endowed with infinite information (like in classical theory represented by real numbers).\\

%In principle, the existence of locally deterministic behaviour and the absence of superluminal signalling, even if the space of SpTs is ontologically deterministic, may be explained by the fundamental inaccessibility of the underlying infinite information. In this view, infinite information could in principle determine all events with perfect precision, but physical constraints prevent observers from accessing it. \textcolor{olive}{The resulting indeterminacy is still ontological, however, not arising particularly from incompatible measurements \cite{del2025features}. Hence, possibly \textit{not} quantum.}\footnote{We may express this distinction using the \textit{Manifest} domain (like outcomes of measurements, settings, physical data, and records in the environment), and the \textit{non-Manifest} domain (where observables, such as position, are defined) to extract specific data  \cite{Fankhauser2022ObsPredict}. In the same philosophy, a \textit{Theory} then serves as the bridge between these two domains. In this paper, we argue that this bridging \textit{theory} is \textit{not} Quantum theory.}\\

We hope this work motivates further investigations into the relation between spacetime symmetries, finite information, causality, and emergent temporal structure. Whether superluminal extensions ultimately correspond to physical reality or remain purely theoretical constructions, they continue to provide a useful arena in which the conceptual limits of relativistic and informational frameworks can be explored with greater clarity.

%%%%%%%%%%%%%%%%%%%%%%%%%%%%%%%%%%%%%%%%%%%%%%%%%%%%%%%%%%%%%%%%%%%%%%%%%%%%%%%%%%%%%%%%%%%%%%%%%%%%%%%%%%%%%%%%%%%
\section{Acknowledgements}

 AS acknowledges the IRA Programme, project no. FENG.02.01-IP.05 0006/23, financed by the FENG program 2021-2027, Priority FENG.02, Measure FENG.02.01., with the support of the FNP. We would also like to thank Matthias Salzger and Sebastian Horvat for their valuable comments and suggestions.

\blk

\bibliography{biblio}

@article{del2021relativity,
  title={The relativity of indeterminacy},
  author={Del Santo, Flavio and Gisin, Nicolas},
  journal={Entropy},
  volume={23},
  number={10},
  pages={1326},
  year={2021},
DOI={https://doi.org/10.3390/e23101326},
  publisher={MDPI}
}

@article{weinert2025propensity,
  title={The Propensity Approach to Quantum Systems},
  author={Weinert, Friedel},
  year={2025}
}

@article{popescu1994quantum,
  title={Quantum nonlocality as an axiom},
  author={Popescu, Sandu and Rohrlich, Daniel},
  journal={Foundations of Physics},
  volume={24},
  number={3},
  pages={379--385},
  year={1994},
  publisher={Springer}
}

@book{eddington2019nature,
  title={The nature of the physical world: The Gifford  lectures 1927},
  author={Eddington, Arthur},
  volume={23},
  year={2019},
  publisher={BoD--Books on Demand}
}

@article{humphreys1985propensities,
  title={Why propensities cannot be probabilities},
  author={Humphreys, Paul},
  journal={The philosophical review},
  volume={94},
  number={4},
  pages={557--570},
  year={1985},
  publisher={JSTOR}
}

@book{suarez2010probabilities,
  title={Probabilities, causes and propensities in physics},
  author={Su{\'a}rez, Mauricio},
  volume={347},
  year={2010},
  publisher={Springer Science \& Business Media}
}

@article{del2023prop,
  title={Potentiality realism: a realistic and indeterministic physics based on propensities},
  author={Del Santo, Flavio and Gisin, Nicolas},
  journal={Euro Jnl Phil Sci},
  volume={13},
  number={58},
  year={2023},
DOI={https://doi.org/10.1007/s13194-023-00561-6},
  publisher={Springer}
}

@article{del2025features,
  title={Which features of quantum physics are not fundamentally quantum but are due to indeterminism?},
  author={Del Santo, Flavio and Gisin, Nicolas},
  journal={Quantum},
  volume={9},
  pages={1686},
  year={2025},
  publisher={Verein zur F{\"o}rderung des Open Access Publizierens in den Quantenwissenschaften},
DOI={	https://doi.org/10.22331/q-2025-04-03-1686}
}

@article{norton2003causation,
  title={Causation as folk science},
  author={Norton, John},
  journal={Philosophia Mathematica},
  volume={3},
  number={4},
  pages={1--22},
  year={2003},
DOI={http://hdl.handle.net/2027/spo.3521354.0003.004}
}

@article{eyink2020renormalization,
  title={Renormalization group approach to spontaneous stochasticity},
  author={Eyink, Gregory and Bandak, Dmytro},
  journal={Physical Review Research},
  volume={2},
  number={4},
  pages={043161},
  year={2020},
DOI={https://doi.org/10.1103/PhysRevResearch.2.043161},
  publisher={APS}
}

@article{chiribella,
  title = {Bell Nonlocality in Classical Systems Coexisting with Other System Types},
  author = {Chiribella, Giulio and Giannelli, Lorenzo and Scandolo, Carlo Maria},
  journal = {Phys. Rev. Lett.},
  volume = {132},
  issue = {19},
  pages = {190201},
  numpages = {6},
  year = {2024},
  month = {May},
  publisher = {American Physical Society},
  doi = {10.1103/PhysRevLett.132.190201},
  url = {https://link.aps.org/doi/10.1103/PhysRevLett.132.190201}
}

@incollection{drossel2015relation,
  title={On the relation between the second law of thermodynamics and classical and quantum mechanics},
  author={Drossel, Barbara},
  booktitle={Why More Is Different: Philosophical Issues in Condensed Matter Physics and Complex Systems},
  pages={41--54},
  year={2015},
DOI={https://doi.org/10.1007/978-3-662-43911-1_3},
  publisher={Springer}
}

@Inbook{del2021indeterminism,
author="Del Santo, Flavio",
title="Indeterminism, Causality and Information: Has Physics Ever Been Deterministic?",
bookTitle="Undecidability, Uncomputability, and Unpredictability",
year="2021",
publisher="Springer International Publishing",
address="Cham",
pages="63--79",
isbn="978-3-030-70354-7",
doi="10.1007/978-3-030-70354-7_5",
url="https://doi.org/10.1007/978-3-030-70354-7_5"
}

@article{ben2020structure,
  title={The Structure of Space and Time, and Physical Indeterminacy},
  author={Ben-Yami, Hanoch},
  journal={arXiv preprint arXiv:2005.05121},
  year={2020},
DOI={
https://doi.org/10.48550/arXiv.2005.05121
Focus to learn more
}
}

@article{einstein1907uber,
  author    = {Albert Einstein},
  title     = {{\"U}ber das Relativit{\"a}tsprinzip und die aus demselben gezogenen Folgerungen},
  journal   = {Jahrbuch der Radioaktivit{\"a}t und Elektronik},
  volume    = {4},
  pages     = {411--462},
  year      = {1907}
}

@article{tolman1917velocities,
  author    = {Richard C. Tolman},
  title     = {The Theory of the Relativity of Motion},
  journal   = {University of California Press},
  year      = {1917}
}

@article{bilaniuk1962meta,
  author    = {O.-M. P. Bilaniuk and V. K. Deshpande and E. C. G. Sudarshan},
  title     = {“Meta” Relativity},
  journal   = {American Journal of Physics},
  volume    = {30},
  number    = {10},
  pages     = {718--723},
  year      = {1962},
  doi       = {10.1119/1.1941773}
}

@article{feinberg1967possibility,
  author    = {Gerald Feinberg},
  title     = {Possibility of Faster-Than-Light Particles},
  journal   = {Physical Review},
  volume    = {159},
  number    = {5},
  pages     = {1089--1105},
  year      = {1967},
  doi       = {10.1103/PhysRev.159.1089}
}

@article{recami1986classical,
  author    = {Erasmo Recami},
  title     = {Classical Tachyons and Possible Applications},
  journal   = {Rivista del Nuovo Cimento},
  volume    = {9},
  number    = {6},
  pages     = {1--178},
  year      = {1986},
  doi       = {10.1007/BF02724327}
}

@article{rembielinski2012meta,
  author    = {Jakub Rembieli{\'n}ski},
  title     = {Superluminal Phenomena and the Quantum Preferred Frame},
  journal   = {International Journal of Modern Physics A},
  volume    = {27},
  number    = {25},
  pages     = {1250102},
  year      = {2012},
  doi       = {10.1142/S0217751X12501021}
}

@article{grushka2013tachyon,
  author    = {Emanuel Grushka},
  title     = {Tachyon Mechanics and Causality: A Systematic Approach},
  journal   = {Physics Essays},
  volume    = {26},
  number    = {4},
  pages     = {523--531},
  year      = {2013},
  doi       = {10.4006/0836-1398-26.4.523}
}

@article{vilasini2022impossibility,
  title={Impossibility of superluminal signaling in Minkowski spacetime does not rule out causal loops},
  author={Vilasini, Venkatesh and Colbeck, Roger},
  journal={Physical Review Letters},
  volume={129},
  number={11},
  pages={110401},
  year={2022},
  publisher={APS}}

@article{horodecki2019relativistic,
  title={The relativistic causality versus no-signaling paradigm for multi-party correlations},
  author={Horodecki, Pawe{\l} and Ramanathan, Ravishankar},
  journal={Nature communications},
  volume={10},
  number={1},
  pages={1701},
  year={2019},
  publisher={Nature Publishing Group UK London}
}

@incollection{durr2009bohmian,
  title={Bohmian mechanics},
  author={D{\"u}rr, Detlev and Goldstein, Sheldon and Tumulka, Roderich and Zangh{\'\i}, Nino},
  booktitle={Compendium of quantum physics},
  pages={47--55},
  year={2009},
  publisher={Springer}
}

@article{gisin2020real,
  title={Real numbers are the hidden variables of classical mechanics},
  author={Gisin, Nicolas},
  journal={Quantum Studies: Mathematics and Foundations},
  volume={7},
  number={2},
  pages={197--201},
  year={2020},
DOI={https://doi.org/10.1007/s40509-019-00211-8},
  publisher={Springer}
}

@article{del2019physics,
  title={Physics without determinism: Alternative interpretations of classical physics},
  author={Del Santo, Flavio and Gisin, Nicolas},
  journal={Physical Review A},
  volume={100},
  number={6},
  pages={062107},
  year={2019},
DOI={https://doi.org/10.1103/PhysRevA.100.062107},
  publisher={APS}
}

@article{ballentine2016propensity,
  title={Propensity, probability, and quantum theory},
  author={Ballentine, Leslie E},
  journal={Foundations of physics},
  volume={46},
  pages={973--1005},
  year={2016},
  publisher={Springer},
DOI={https://doi.org/10.1007/s10701-016-9991-0}
}

@article{parker1969,
  title = {Faster-Than-Light Inertial Frames and Tachyons},
  author = {Parker, Leonard},
  journal = {Phys. Rev.},
  volume = {188},
  issue = {5},
  pages = {2287--2292},
  numpages = {0},
  year = {1969},
  month = {Dec},
  publisher = {American Physical Society},
  doi = {10.1103/PhysRev.188.2287},
  url = {https://link.aps.org/doi/10.1103/PhysRev.188.2287}
}

@article{gisin2021indeterminism,
  title={Indeterminism in physics, classical chaos and {B}ohmian mechanics: Are real numbers really real?},
  author={Gisin, Nicolas},
  journal={Erkenntnis},
  volume={86},
  number={6},
  pages={1469--1481},
  year={2021},
DOI={https://doi.org/10.1007/s10670-019-00165-8},
  publisher={Springer}
}

@article{DelSantoGisin2024,
  author        = {Flavio Del Santo and Nicolas Gisin},
  title         = {Creative and geometric times in physics, mathematics, logic, and philosophy},
  year          = {2024},
  journal       = {arXiv preprint arXiv:2404.06566},
  eprint        = {2404.06566},
  archivePrefix = {arXiv},
  primaryClass  = {physics.hist-ph},
  url           = {https://arxiv.org/abs/2404.06566}
}

@article{de2009symplectic,
  title={The symplectic camel and the uncertainty principle: The tip of an iceberg?},
  author={De Gosson, Maurice A.},
  journal={Foundations of Physics},
  volume={39},
  pages={194--214},
  year={2009},
  publisher={Springer},
DOI={https://doi.org/10.1007/s10701-009-9272-2}
}

@article{popper1959propensity,
  title={The propensity interpretation of probability},
  author={Popper, Karl R},
  journal={The British Journal for the Philosophy of Science},
  volume={10},
  number={37},
  pages={25--42},
  year={1959},
  publisher={Oxford University Press},
DOI={https://doi.org/10.1093/bjps/X.37.25}
}

@book{born2012physics,
  title={Physics in my generation},
  author={Born, Max},
  year={2012},
DOI={https://doi.org/10.1007/978-3-662-25189-8},
  publisher={Springer Science \& Business Media}
}

@article{ornstein1989ergodic,
  title={Ergodic Theory, Randomness, and" Chaos"},
  author={Ornstein, Donald S},
  journal={Science},
  volume={243},
  number={4888},
  pages={182--187},
  year={1989},
DOI={https://doi.org/10.1126/science.243.4888.182},
  publisher={American Association for the Advancement of Science}
}

@article{DelSanto_2022,
  doi       = {10.1088/1367-2630/acae3b},
  url       = {https://doi.org/10.1088/1367-2630/acae3b},
  year      = {2023},    
  month     = {dec},        publisher = {IOP Publishing},
  volume    = {24},
  number    = {12},
  pages     = {128001},
  author    = {Del Santo, Flavio and Horvat, Sebastian},
  title     = {Comment on 'Quantum principle of relativity'},
  journal   = {New Journal of Physics}
}

@article{del2018striving,
  title={Striving for realism, not for determinism: Historical misconceptions on {E}instein and {B}ohm},
  author={Del Santo, Flavio},
  journal={APS News},
  year={2018},
nolink={}
}

@article{Horodecki2023,
  author       = {Horodecki, Ryszard},
  title        = {Comment on ‘Quantum principle of relativity’},
  journal      = {New Journal of Physics},
  year         = {2023},
  volume       = {25},
  number       = {12},
  pages        = {128001},
  doi          = {10.1088/1367-2630/ad10ff},
  url          = {https://doi.org/10.1088/1367-2630/ad10ff}
}

@article{Grudka2023,
  author  = {Andrzej Grudka and Jedrzej Stempin and Jan Wójcik and Antoni Wójcik},
  title   = {Superluminal observers do not explain quantum superpositions},
  journal = {Physics Letters A},
  volume  = {487},
  pages   = {129127},
  year    = {2023},
  doi     = {10.1016/j.physleta.2023.129127}
}

@article{dragan2020quantum,
  title={Quantum principle of relativity},
  author={Dragan, Andrzej and Ekert, Artur},
  journal={New Journal of Physics},
  volume={22},
  number={3},
  pages={033038},
  year={2020},
  publisher={IOP Publishing},
DOI={https://doi.org/10.1088/1367-2630/ab76f7}
}

@article{grudka2022,
  title = {Comment on ‘Quantum principle of relativity’},
  author = {Grudka, Andrzej and W{\'o}jcik, Antoni},
  journal = {New Journal of Physics},
  volume = {24},
  number = {9},
  pages = {098001},
  year = {2022},
  month = {Oct},
  publisher = {IOP Publishing},
  doi = {10.1088/1367-2630/ac924e},
  url = {https://doi.org/10.1088/1367-2630/ac924e}
}

@article{gisinpropold,
  title={Propensities in a non-deterministic physics},
  author={Gisin, Nicolas},
  journal={Synthese},
  volume={89},
  pages={287--297},
  year={1991},
  publisher={Springer},
DOI={https://doi.org/10.1007/BF00413910}
}

@article{gisin2021synthese,
  title={Indeterminism in physics and intuitionistic mathematics},
  author={Gisin, Nicolas},
  journal={Synthese},
  volume={199},
  number={5},
  pages={13345--13371},
  year={2021},
DOI={https://doi.org/10.1007/s11229-021-03378-z},
  publisher={Springer}
}

@book{prigogine1997end,
  title={The end of certainty},
  author={Prigogine, Ilya},
  year={1997},
url       = {https://books.google.com/books?id=YP5RAQAAIAAJ},
  publisher={Simon and Schuster}
}

@article{PhysRevD.23.287,
  author = {Bekenstein, Jacob D.},
  title = {Universal upper bound on the entropy-to-energy ratio for bounded systems},
  journal = {Physical Review D},
  volume = {23},
  number = {2},
  pages = {287--298},
  year = {1981},
  doi = {10.1103/PhysRevD.23.287}
}

@article{thooft1993dimensional,
  author = {'t Hooft, Gerard},
  title = {Dimensional Reduction in Quantum Gravity},
  journal = {arXiv preprint gr-qc/9310026},
  year = {1993},
  archivePrefix = {arXiv},
  eprint = {gr-qc/9310026}
}

@article{susskind1995world,
  author = {Susskind, Leonard},
  title = {The World as a Hologram},
  journal = {Journal of Mathematical Physics},
  volume = {36},
  number = {11},
  pages = {6377--6396},
  year = {1995},
  doi = {10.1063/1.531249},
  archivePrefix = {arXiv},
  eprint = {hep-th/9409089}
}

@article{bousso2002holographic,
  author = {Bousso, Raphael},
  title = {The holographic principle},
  journal = {Reviews of Modern Physics},
  volume = {74},
  pages = {825--874},
  year = {2002},
  doi = {10.1103/RevModPhys.74.825},
  archivePrefix = {arXiv},
  eprint = {hep-th/0203101}
}

@article{srednicki1993entropy,
  author = {Srednicki, Mark},
  title = {Entropy and area},
  journal = {Physical Review Letters},
  volume = {71},
  number = {5},
  pages = {666--669},
  year = {1993},
  doi = {10.1103/PhysRevLett.71.666},
  archivePrefix = {arXiv},
  eprint = {hep-th/9303048}
}

@article{harlow2016jerusalem,
  author = {Harlow, Daniel},
  title = {Jerusalem Lectures on Black Holes and Quantum Information},
  journal = {Reviews of Modern Physics},
  volume = {88},
  number = {1},
  pages = {015002},
  year = {2016},
  doi = {10.1103/RevModPhys.88.015002},
  archivePrefix = {arXiv},
  eprint = {1409.1231}
}

@article{bekenstein2003information,
  author = {Bekenstein, Jacob D.},
  title = {Information in the Holographic Universe},
  journal = {Scientific American},
  volume = {289},
  number = {2},
  pages = {58--65},
  year = {2003},
nolink={}
}

@article{santo2023open,
  title={The open past in an indeterministic physics},
  author={Del Santo, Flavio and Gisin, Nicolas},
  journal={Foundations of Physics},
  volume={53},
  number={1},
  pages={4},
  year={2023},
DOI={https://doi.org/10.1007/s10701-022-00645-y},
  publisher={Springer}
}

@article{pironio2010random,
  title={Random numbers certified by {B}ell’s theorem},
  author={Pironio, Stefano and Ac{\'\i}n, Antonio and Massar, Serge and de La Giroday, A Boyer and Matsukevich, Dzmitry N and Maunz, Peter and Olmschenk, Steven and Hayes, David and Luo, Le and Manning, T Andrew and others},
  journal={Nature},
  volume={464},
  number={7291},
  pages={1021--1024},
  year={2010},
  publisher={Nature Publishing Group UK London},
DOI={https://doi.org/10.1038/nature09008}
}

@article{Vilasini2024,
  author = {V. Vilasini and Renato Renner},
  title = {Fundamental limits for realizing quantum processes in spacetime},
  journal = {Physical Review Letters},
  volume = {133},
  number = {8},
  pages = {080201},
  year = {2024},
  doi = {10.1103/PhysRevLett.133.080201}
}

@article{Vilasini2024a,
  author = {V. Vilasini and Renato Renner},
  title = {Embedding cyclic information-theoretic structures in acyclic spacetimes: No-go results for indefinite causality},
  journal = {Physical Review A},
  volume = {110},
  number = {2},
  pages = {022227},
  year = {2024},
  doi = {10.1103/PhysRevA.110.022227}
}

@article{vilasini2022general,
  title = {General framework for cyclic and fine-tuned causal models and their compatibility with space-time},
  author = {Vilasini, V. and Colbeck, Roger},
  journal = {Phys. Rev. A},
  volume = {106},
  issue = {3},
  pages = {032204},
  numpages = {43},
  year = {2022},
  month = {Sep},
  publisher = {American Physical Society},
  doi = {10.1103/PhysRevA.106.032204},
  url = {https://link.aps.org/doi/10.1103/PhysRevA.106.032204}
}

@article{sen2026superluminal,
  title={Superluminal quantum reference frames},
  author={Sen, Amrapali and Salzger, Matthias and Rudnicki, {\L}ukasz},
  journal={New Journal of Physics},
  volume={28},
  number={4},
  pages={044505},
  year={2026},
  publisher={IOP Publishing}
}

@article{Schwartz2018,
  author    = {Charles Schwartz},
  title     = {Tachyon dynamics—for neutrinos?},
  journal   = {International Journal of Modern Physics A},
  volume    = {33},
  number    = {18},
  pages     = {1850056},
  year      = {2018},
  doi       = {10.1142/S0217751X18500561}
}

@article{Einstein1916,
  author    = {Albert Einstein},
  title     = {Die Grundlage der allgemeinen Relativitätstheorie},
  journal   = {Annalen der Physik},
  volume    = {354},
  number    = {7},
  pages     = {769--822},
  year      = {1916},
  doi       = {10.1002/andp.19163540702}
}

@article{Einstein1923Foundation,
  author    = {Albert Einstein},
  title     = {The Foundation of the General Theory of Relativity},
  journal   = {Methuen and Company},
  year      = {1923},
  note      = {English translation of Einstein's 1916 paper}
}

@article{Einstein1914,
  author    = {Albert Einstein},
  title     = {Die formale Grundlage der allgemeinen Relativitätstheorie},
  journal   = {Sitzungsberichte der Königlich Preußischen Akademie der Wissenschaften (Berlin)},
  pages     = {1030--1085},
  year      = {1914}
}

@article{Stachel1980HoleArgument,
  author    = {John Stachel},
  title     = {Einstein's Search for General Covariance, 1912--1915},
  journal   = {Einstein and the History of General Relativity},
  volume    = {1},
  pages     = {63--100},
  year      = {1989},
  publisher = {Birkhäuser}
}

@article{Sorkin2003,
  author    = {Rafael D. Sorkin},
  title     = {Causal Sets: Discrete Gravity},
  journal   = {Lectures on Quantum Gravity},
  pages     = {305--327},
  year      = {2005},
  archivePrefix = {arXiv},
  eprint    = {gr-qc/0309009},
  doi       = {10.1007/0-387-24999-2_10}
}

@article{Dowker2005,
  author    = {Fay Dowker},
  title     = {Causal Sets and the Deep Structure of Spacetime},
  journal   = {100 Years of Relativity: Space-Time Structure},
  pages     = {445--464},
  year      = {2005},
  archivePrefix = {arXiv},
  eprint    = {gr-qc/0508109},
  doi       = {10.1142/9789812704030_0017}
}

@article{Bombelli1987,
  author    = {Luca Bombelli and Joohan Lee and David Meyer and Rafael D. Sorkin},
  title     = {Space-Time as a Causal Set},
  journal   = {Physical Review Letters},
  volume    = {59},
  number    = {5},
  pages     = {521--524},
  year      = {1987},
  doi       = {10.1103/PhysRevLett.59.521}
}

@article{RidleyAdlam2023,
  author = {Michael Ridley and Emily Adlam},
  title = {Time and Event Symmetry in Quantum Mechanics},
  journal = {arXiv preprint arXiv:2312.13524},
  year = {2023},
  archivePrefix = {arXiv},
  eprint = {2312.13524},
  primaryClass = {quant-ph}
}

@article{parker1969ftl,
  author = {Leonard Parker},
  title = {Faster-Than-Light Inertial Frames and Tachyons},
  journal = {Physical Review},
  volume = {188},
  pages = {2287--2292},
  year = {1969},
  doi = {10.1103/PhysRev.188.2287}
}

@article{ramon1980complex,
  author = {Ceon Ramon and Elizabeth A. Rauscher},
  title = {Superluminal Transformations in Complex Minkowski Spaces},
  journal = {Foundations of Physics},
  volume = {10},
  number = {7-8},
  pages = {661--669},
  year = {1980},
  doi = {10.1007/BF00715047}
}

@article{maccarrone1984revisitation,
  author = {G. D. Maccarrone and Erasmo Recami},
  title = {The Introduction of Superluminal Lorentz Transformations: A Revisitation},
  journal = {Foundations of Physics},
  volume = {14},
  number = {5},
  pages = {367--407},
  year = {1984},
  doi = {10.1007/BF00738808}
}

@article{maccarrone1982group,
  author = {G. D. Maccarrone and Erasmo Recami},
  title = {Revisiting the Superluminal Lorentz Transformations and Their Group-Theoretical Properties},
  journal = {Lettere al Nuovo Cimento},
  volume = {34},
  pages = {251--256},
  year = {1982},
  doi = {10.1007/BF02817120}
}

@article{ibison2007tachyons,
  author = {Michael Ibison},
  title = {Tachyons and Superluminal Boosts},
  journal = {arXiv preprint arXiv:0704.3277},
  year = {2007}
}

@article{roldan2023superluminal,
  author = {Diego Roldán and Francisco Roldán-Aráuz},
  title = {A Transformation Factor for Superluminal Motion That Preserves Symmetrically the Spacetime Intervals},
  journal = {Symmetry},
  volume = {15},
  number = {6},
  pages = {1177},
  year = {2023},
  doi = {10.3390/sym15061177}
}

\end{document}